\documentclass[11pt]{article}
\usepackage{pgfplots}
\usepackage{fullpage}
\pgfplotsset{compat=1.7}
\usepackage{subcaption}
\usepackage[T1]{fontenc}
\usepackage{lmodern}
\usepackage{amsmath}
\usepackage{amssymb}
\usepackage{amsthm}
\usepackage{color,soul}
\usepackage{xcolor}
\usepackage{graphicx}
\usepackage{tcolorbox}
\usepackage[normalem]{ulem}
\usepackage{float}
\usepackage{thmtools}
\usepackage{hyperref}
\usepackage{cleveref}
\usepackage{microtype}
\usepackage{caption}
\usepackage{bbm}
\usepackage{hyperref, color}
\hypersetup{colorlinks=true,citecolor=blue, linkcolor=blue, urlcolor=blue}
\usepackage[linesnumbered,boxed,ruled,vlined]{algorithm2e}
\usepackage{bm}
\usepackage{bbm}
\usepackage[numbers]{natbib}
\usepackage{xcolor}
\usepackage{enumerate} 
\usepackage{enumitem}
\usepackage{tabularx}
\usepackage{array}
\usepackage{cleveref}
\usepackage{tcolorbox}
\usepackage[thinlines]{easytable}
\usepackage{tikz}

\usetikzlibrary{decorations.pathmorphing, arrows.meta}
\newcolumntype{L}[1]{>{\raggedright\arraybackslash}p{#1}}
\newcolumntype{C}[1]{>{\centering\arraybackslash}m{#1}}
\newcolumntype{R}[1]{>{\raggedleft\arraybackslash}p{#1}}

\usepackage{makecell}
\usepackage{footnote}
\usepackage[ruled, vlined]{algorithm2e}
\makesavenoteenv{tabular}

\newcommand{\DTV}[2]{d_{\mathrm{TV}}\left({#1},{#2}\right)}

\renewcommand{\epsilon}{\varepsilon}

\newtheorem{theorem}{Theorem}[section]
\newtheorem{observation}[theorem]{Observation}

\newtheorem*{claim*}{Claim}

\newtheorem{fact}[theorem]{Fact}
\newtheorem{lemma}[theorem]{Lemma}

\theoremstyle{definition}

\newtheorem{definition}[theorem]{Definition}
\newtheorem{problem}[theorem]{Problem}
\newtheorem{remark}[theorem]{Remark}
\newtheorem*{remark*}{Remark}
\newtheorem{assumption}[theorem]{Assumption}

\def\Pr{\mathop{\mathbf{Pr}}\nolimits}

\renewcommand{\emptyset}{\varnothing}

 \newcommand{\tuple}[1]{\left(#1\right)} 
 
 \newcommand{\tp}{\tuple}

\newcommand{\defeq}{\triangleq}

\def\*#1{\mathbf{#1}} 
\def\+#1{\mathcal{#1}} 
\def\-#1{\mathrm{#1}} 
\def\^#1{\mathbb{#1}} 

\usepackage{todonotes}
\usepackage{xifthen}

\renewcommand{\Pr}[2][]{ \ifthenelse{\isempty{#1}}
  {\mathbf{Pr}\left[#2\right]} {\mathbf{Pr}_{#1}\left[#2\right]} } 
\newcommand{\E}[2][]{ \ifthenelse{\isempty{#1}}
  {\mathbf{\mathbf{E}}\left[#2\right]}
  {\mathbf{\mathbf{E}}_{#1}\left[#2\right]} }
  \newcommand{\Var}[2][]{ \ifthenelse{\isempty{#1}}
  {\mathbf{\mathbf{Var}}\left[#2\right]}
  {\mathbf{\mathbf{Var}}_{#1}\left[#2\right]} }

\crefname{theorem}{Theorem}{Theorems}
\crefname{observation}{Observation}{Observations}
\crefname{claim}{Claim}{Claims}
\crefname{condition}{Condition}{Conditions}
\crefname{algorithm}{Algorithm}{Algorithms}
\crefname{property}{Property}{Properties}
\crefname{example}{Example}{Examples}
\crefname{fact}{Fact}{Facts}
\crefname{lemma}{Lemma}{Lemmas}
\crefname{corollary}{Corollary}{Corollaries}
\crefname{definition}{Definition}{Definitions}
\crefname{remark}{Remark}{Remarks}
\crefname{proposition}{Proposition}{Propositions}
\crefname{equation}{equation}{equations}
\crefname{enumi}{}{}
\crefname{enumii}{}{}
\crefname{enumiii}{}{}
\crefname{enumiv}{}{}
\creflabelformat{enumi}{#2#1#3}
\creflabelformat{enumii}{#2#1#3}
\creflabelformat{enumiii}{#2#1#3}
\creflabelformat{enumiv}{#2#1#3}

\newboolean{DoubleBlind}
\setboolean{DoubleBlind}{false}

\title{On Computing Total Variation Distance \\ Between Mixtures of Product Distributions}

\author{\ifthenelse{\boolean{DoubleBlind}}{Author(s)}
{Weiming Feng \footnote{School of Computing and Data Science, The University of Hong Kong. \\\makebox[1.78em][l]{}Emails: \texttt{wfeng@hku.hk}, \texttt{fyc0130@connect.hku.hk}, and \texttt{ymjessen02@connect.hku.hk}.
} \and Yucheng Fu \footnotemark[1] \and Minji Yang \footnotemark[1] \and Anqi Zhang \footnote{The Institute for Interdisciplinary Information Sciences, Tsinghua University. \\ \makebox[1.78em][l]{}Email: \texttt{zhangaq21@mails.tsinghua.edu.cn}.}
}}

\date{}

\begin{document}

\maketitle

\begin{abstract}
    We study the problem of approximating the total variation distance between two mixtures of product distributions over an $n$-dimensional discrete domain. Given two mixtures $\^P$ and $\^Q$ with $k_1$ and $k_2$ product distributions over $[q]^n$, respectively, we give a randomized algorithm that approximates $\DTV{\^P}{\^Q}$ within a multiplicative error of $(1\pm \varepsilon)$ in time $\mathrm{poly}((nq)^{k_1+k_2},1/\epsilon)$. 
    We also study the special case of mixtures of Boolean subcubes over $\{0,1\}^n$. For this class, we give a deterministic algorithm that exactly computes the total variation distance in time $\mathrm{poly}(n,2^{O(k_1+k_2)})$, and show that exact computation is $\#\mathsf{P}$-hard when $k_1+k_2=\Theta(n)$.
\end{abstract}

\section{Introduction}
Let $\^P$ and $\^Q$ be two discrete distributions over a sample space $\Omega$. The \emph{total variation distance (TV-distance)} between $\^P$ and $\^Q$ is defined as
\begin{align*}
    \DTV{\^P}{\^Q} = \frac{1}{2} \sum_{\omega \in \Omega} \left| \^P(\omega) - \^Q(\omega) \right|  =  \sum_{\omega \in \Omega}\max\{0, \^P(\omega) - \^Q(\omega)\}.
\end{align*}
The TV-distance is one of the most fundamental metrics for measuring the difference between distributions, as it characterizes the optimal distinguishability between $\^P$ and $\^Q$. Computing the TV-distance is therefore a basic problem arising in learning and testing.

A line of research has investigated the problem of \emph{approximating} the TV-distance between two high-dimensional distributions, which is particularly interesting when the two distributions admit \emph{succinct representations}. 
Early works~\cite{sahai2003complete,CanonneR14,ChenK14,Kiefer18,BGMV20} studied the algorithms and complexity of approximating the TV-distance with additive error.
More recently, \cite{BGMMPV23} established that even for two product distributions, the \emph{exact computation} of the TV-distance is $\#\mathsf{P}$-hard.
In the same paper, the authors initiated the study of approximating the TV-distance with \emph{relative error}, a more challenging task than the classical additive-error approximation. Subsequently, polynomial-time approximation algorithms have been developed for product distributions~\cite{BGMMPV23, FGJW23,FengLL24}, as well as for many high-dimensional distributions defined by \emph{local interactions}, including Markov chains~\cite{FengLL24}, Bayesian networks~\cite{BGMMPV24ICML}, and undirected graphical models~\cite{feng2025approximating}.

Beyond local interactions, it is natural to consider more general high-dimensional distributions with inherently \emph{non-local structure}.
A fundamental and extensively studied example is given by \emph{mixtures of product distributions}.
Let $[q]=\{1,2,\ldots,q\}$ and let $\Omega = [q]^n$ denote the $n$-dimensional sample space.
Fix two integers $k_1, k_2 \geq 1$. For each $s \in [k_1]$, let $\^P^{(s)} = \bigotimes_{i=1}^n \^P_i^{(s)}$ be a product distribution over $\Omega$, where $\^P_i^{(s)}$ denotes the marginal on the $i$-th coordinate, so that the coordinates are mutually independent under $\^P^{(s)}$. Let $\alpha^{(1)},\alpha^{(2)},\ldots,\alpha^{(k_1)}$ be mixing weights satisfying $\sum_{s=1}^{k_1} \alpha^{(s)} = 1$. The resulting mixture distribution is written as 
$$\^P = \sum_{s=1}^{k_1} \alpha^{(s)} \bigotimes_{i=1}^n \^P^{(s)}_i.$$
More explicitly, $\^P$ is the distribution over $\Omega = [q]^n$ given by
\begin{align}\label{equation:mixture-of-product-distributions}
\forall \omega \in \Omega, \quad  \^P(\omega) = \sum_{s=1}^{k_1} \alpha^{(s)} \^P^{(s)}(\omega) = \sum_{s=1}^{k_1} \alpha^{(s)} \prod_{i=1}^n \^P_i^{(s)}(\omega_i),
\end{align}

Mixtures of product distributions form a fundamental class of \emph{latent-variable models} and have been extensively studied in learning theory~\cite{FeldmanOS08,JainO14,ChenM19,GordonMRS21,GordonJMRS24}.
To draw a random sample $X \sim \^P$, one first samples a \emph{hidden} component index $S \in [k_1]$ according to $\Pr[]{S=s} = \alpha^{(s)}$, and then samples $X \sim \^P^{(S)} = \bigotimes_{i=1}^n \^P^{(S)}_i$ from the product distribution $\^P^{(S)}$.
Although each component $\^P^{(s)}$ is itself a product distribution, the hidden index $S$ induces a highly \emph{non-local dependency} that simultaneously couples all coordinates of the sample $X = (X_1, X_2, \ldots, X_n)$, which cannot be captured by any sparse local structure.

In this paper, we study the problem of \emph{approximating} the TV-distance between two mixtures of product distributions to within an arbitrarily small relative error, formally stated as follows. 

\begin{problem}\label{problem:approx-tv-distance}
Approximate the TV-distance between two mixtures of product distributions.
\begin{itemize}
    \item \textbf{Input}: An error bound $\epsilon > 0$ and two mixtures of product distributions over $\Omega = [q]^n$: 
    \begin{align*}
        \^P = \sum_{s=1}^{k_1} \alpha^{(s)} \bigotimes_{i=1}^n \^P^{(s)}_i \text{ and } \^Q = \sum_{t=1}^{k_2} \beta^{(t)} \bigotimes_{i=1}^n \^Q^{(t)}_i.
    \end{align*}
    The distribution $\^P$ is specified by the $k_1$ mixing weights $\alpha^{(1)},\alpha^{(2)},\ldots,\alpha^{(k_1)}$ together with the $k_1$ product distributions $\^P^{(1)},\^P^{(2)},\ldots,\^P^{(k_1)}$, where each $\^P^{(s)}$ is described by its $n$ marginals $\^P_1^{(s)}, \^P_2^{(s)}, \ldots, \^P_n^{(s)}$ over $[q]$. The distribution $\^Q$ is given analogously. The total input size of the two mixtures is $O((k_1+k_2)nq)$.
    \item \textbf{Output}: A number $\hat{d}$ such that
    \begin{align*}
       (1-\epsilon) \DTV{\^P}{\^Q} \leq \hat{d} \leq (1+\epsilon) \DTV{\^P}{\^Q}.
    \end{align*}
\end{itemize}
\end{problem}

Due to the non-local nature of these dependencies among all coordinates, existing techniques for TV-distance approximation cannot be directly applied to mixtures of product distributions. The only prior progress is a polynomial-time algorithm for \emph{deciding} whether $\DTV{\^P}{\^Q} = 0$ or $\DTV{\^P}{\^Q} > 0$~\cite{0001GMMPV25}, leaving the approximation problem open.
We partially resolve this open question by giving an FPRAS for the TV-distance between two mixtures of product distributions, provided that the numbers of components $k_1$ and $k_2$ are constants.

\begin{theorem}\label{thm:tv-mix-prod}
Let $k_1,k_2 \geq 1$ be constants. There exists an FPRAS for \Cref{problem:approx-tv-distance}: given any $\epsilon > 0$ and any pair of mixtures of product distributions $\^P,\^Q$ over $[q]^n$, where $\^P = \sum_{s=1}^{k_1} \alpha^{(s)} \bigotimes_{i=1}^n \^P^{(s)}_i$ has $k_1$ components and $\^Q = \sum_{t=1}^{k_2} \beta^{(t)} \bigotimes_{i=1}^n \^Q^{(t)}_i$ has $k_2$ components, it solves the problem in time $O(\frac{q^2(nq)^{2(k_1+k_2-1)}}{\varepsilon^2})$ with success probability at least $99\%$.
\end{theorem}

Our algorithm is a Monte Carlo algorithm based on a carefully designed estimator, and we provide an overview of it in \Cref{sec:overview}. Since the number of components $k_1+k_2$ appears as an exponent in the running time, we require $k_1+k_2$ to be a constant. A similar dependence on the number of components also arises in the running time of certain learning tasks for mixtures of product distributions~\cite{FeldmanOS08,ChenM19,GordonMRS21}. We leave it as an open question to understand the computational complexity when the number of components grows with $n$.

Our second result concerns mixtures of Boolean subcubes, a special case of mixtures of product distributions. A distribution $\^P = \sum_{s=1}^k \alpha^{(s)} \bigotimes_{i=1}^n \^P_i^{(s)}$ over $\{0,1\}^n$ is called a mixture of Boolean subcubes if $\^P_i^{(s)}(1) \in \{0,\frac{1}{2},1\}$ for every $s \in [k]$ and $i \in [n]$, where $\^P_i^{(s)}$ denotes the marginal distribution of the $i$-th coordinate in the $s$-th component. Equivalently, for each component $s$, the product distribution $\^P^{(s)}$ fixes every coordinate $i$ with $\^P_i^{(s)}(1) = 1$ to $1$, fixes every coordinate $i$ with $\^P_i^{(s)}(1) = 0$ to $0$, and is uniform on the remaining coordinates. Mixtures of Boolean subcubes have been extensively studied in learning theory~\cite{ChenM19,BlancLMT23,BlancLST25} and capture many natural distributions; for instance, the uniform distribution over the satisfying assignments of a decision tree is a mixture of Boolean subcubes~\cite{FeldmanOS08}. 



In contrast to the general case, where we only obtain an FPRAS when $k_1+k_2 = O(1)$, we show that the TV-distance between two mixtures of Boolean subcubes $\^P$ and $\^Q$ can be computed \emph{exactly} in time polynomial in the dimension $n$, provided that $k_1+k_2 = O(\log n)$.

\begin{problem}\label{problem:compute-subcube}
    Compute exactly the total variation distance between two mixtures of Boolean subcubes $\^P = \sum_{s=1}^{k_1} \alpha^{(s)} \bigotimes_{i=1}^n \^P^{(s)}_i$ and $\^Q = \sum_{t=1}^{k_2} \beta^{(t)} \bigotimes_{i=1}^n \^Q^{(t)}_i$.
    \begin{itemize}
        \item \textbf{Input}: two mixtures of Boolean subcubes, given in the same format as in \Cref{problem:approx-tv-distance}.
        \item \textbf{Output}: the \emph{exact} value of $\DTV{\^P}{\^Q}$
    \end{itemize}
\end{problem}

\begin{theorem}\label{thm:tv-subcube}
    There exists a deterministic algorithm for \Cref{problem:compute-subcube}: 
    given any pair of mixtures of Boolean subcubes $\^P,\^Q$ over $\{0,1\}^n$, where $\^P = \sum_{s=1}^{k_1} \alpha^{(s)} \bigotimes_{i=1}^n \^P^{(s)}_i$ has $k_1$ components and $\^Q = \sum_{t=1}^{k_2} \beta^{(t)} \bigotimes_{i=1}^n \^Q^{(t)}_i$ has $k_2$ components, it solves the problem in $O((k_1+k_2)3^{k_1+k_2} \cdot n)$ time.
\end{theorem}

The running time in \Cref{thm:tv-subcube} is polynomial in $n$ whenever $k_1+k_2 = O(\log n)$, yielding a deterministic polynomial-time algorithm in this regime.
This does not contradict the $\#\mathsf{P}$-hardness of exact TV-distance computation established in~\cite{BGMMPV23}: their hardness already holds for two product distributions with arbitrary marginals, whereas our algorithm crucially exploits the restriction that each marginal of a Boolean subcube takes value in $\{0,\frac{1}{2},1\}$.

The idea of the algorithm is as follows.
For each component $\^P^{(s)}$ (and similarly for each $\^Q^{(t)}$), the value of $\^P^{(s)}(\omega)$ is either 0 or $2^{-r}$, where $r$ is the number of coordinates $i \in [n]$ with $\^P_i^{(s)}(1) = \frac{1}{2}$. Indeed, if some coordinate is fixed to be 1 (or 0) but $\omega_i=0$ (or 1), then $\^P^{(s)}(\omega)=0$; otherwise, each free coordinate contributes a factor of $1/2$, so $\^P^{(s)}(\omega)=2^{-r}$.
Thus, although there are $2^n$ configurations $\omega \in \{0,1\}^n$, the possible values of $|\^P(\omega)-\^Q(\omega)|$ are at most $2^{k_1+k_2}$.
Our algorithm enumerates these possible values of $|\^P(\omega)-\^Q(\omega)|$  and counts the number of configurations $\omega \in \{0,1\}^n$ that achieve each one.
The details are given in~\Cref{sec:algorithm-subcube}.


Finally, we complement this positive result by showing that the dependence on the number of components cannot be removed entirely: when the number of components grows linearly with $n$, exact computation becomes $\#\mathsf{P}$-hard.

\begin{theorem}\label{thm:hardness-subcube}
    If the number of mixture components is $k_1 + k_2 = \Theta(n)$, then the exact computation of the total variation distance in \Cref{problem:compute-subcube}
  is $\#\mathsf{P}$-hard.
\end{theorem}

The theorem is proved by a reduction showing that if we can exactly compute the TV-distance between two mixtures of Boolean subcubes with $k_1+k_2=\Theta(n)$ in polynomial time, then we can also exactly compute $\#\mathsf{3SAT}$ in polynomial time.
The detailed proof is in \Cref{sec:proof-hardness-subcube}.

\paragraph{Open problems}
As discussed after \Cref{thm:tv-mix-prod}, the dependence on the number of components in the running time is not yet well understood. \cite{0001GMMPV25} gave an algorithm that decides whether $\DTV{\^P}{\^Q}=0$ in time $\mathrm{poly}(n,k_1+k_2)$, whereas the running time of our approximation algorithm depends exponentially on $k_1+k_2$. The analogous regime $k_1+k_2 = \omega(1)$ is also not well understood for learning tasks on mixtures of product distributions, where obtaining algorithms with polynomial dependence on the number of components remains a central challenge~\cite{FeldmanOS08,ChenM19,GordonMRS21}. A natural question is whether one can approximate $\DTV{\^P}{\^Q}$ in time polynomial in $k_1+k_2$, or prove a hardness result ruling this out when the number of components grows with $n$. As a first step toward the latter, \Cref{thm:hardness-subcube} already shows that exact computation is $\#\mathsf{P}$-hard when $k_1+k_2 = \Theta(n)$, and one direction is to extend its proof to rule out approximation as well.


One can ask analogous TV-distance approximation questions for other structured mixture models. Natural candidates include mixtures of Gaussians~\cite{AshtianiBHLPM18,LiS17}, mixtures of Markov chains and MDPs~\cite{KausikTT23}. These models also give succinct descriptions of high-dimensional distributions and are widely studied in machine learning, but their dependence structures are more complex than that of mixtures of product distributions. More generally, it would be interesting to understand the complexity of computing or approximating the TV-distance between trajectory distributions induced by reward-mixing Markov decision processes (RMMDPs). RMMDPs use latent mixture structure to model the reward context, and TV-distance bounds between trajectory distributions have been used to learn near-optimal policies~\cite{KwonECM23}. The problem of computing the TV-distance between RMMDPs appears to require new ideas.

At last, our algorithm for mixtures of product distributions is randomized. It remains open to design a deterministic algorithm for approximating the TV-distance between mixtures of product distributions. Currently, deterministic algorithms are only known for approximating the TV-distance between product distributions~\cite{FengLL24,BGMMPV23}.

\section{Technical Overview}\label{sec:overview}
We present the main ideas behind the algorithm in \Cref{thm:tv-mix-prod}.
An obstacle for approximating the TV-distance is that the TV-distance itself can be very small. 
To see the reason,
for two distributions $\^P$ and $\^Q$ over the sample space $\Omega$, where $\^Q$ is absolutely continuous w.r.t.\,$\^P$, one can write
\begin{align*}
    \DTV{\^P}{\^Q}=\sum_{\omega \in \Omega}\^P(\omega)\max\left\{0,1-\frac{\^Q(\omega)}{\^P(\omega)}\right\}=\E[\omega\sim \^P]{g(\omega)},
\end{align*}
where $g(\omega):=\max\left\{0,1-\frac{\^Q(\omega)}{\^P(\omega)}\right\}$. We can efficiently sample $\omega$ and compute $g(\omega)$, so that a direct Monte Carlo estimate of $\DTV{\^P}{\^Q} = \E[\omega\sim \^P]{g(\omega)}$ can only achieve \emph{additive} error in polynomial time. However, since the expectation $\E[\omega\sim \^P]{g(\omega)}$ itself can be exponentially small, this simple approach is not sufficient to achieve a \emph{relative} error approximation required by \Cref{problem:approx-tv-distance}.

\subsection{Approximating the TV-distance via a coarse coupling}\label{subsection:dtv-from-close-coupling}
A framework for approximating the TV-distance was proposed in \cite{FGJW23} for two product distributions (i.e., $k_1=k_2=1$ in our setting). It can be generalized to any two distributions $\^P$ and $\^Q$.
Instead of approximating $\DTV{\^P}{\^Q}$ directly, they estimated the normalized quantity $\DTV{\^P}{\^Q} / \tilde{d}$, where $\tilde{d}$ is a polynomial-time computable \emph{coarse} estimator. Moreover, if the ratio $\DTV{\^P}{\^Q} / \tilde{d}$ is lower bounded by $1/\mathrm{poly}(n)$, one can try to achieve a relative-error approximation to the 
ratio in polynomial time. 
Hence, we can recover $\DTV{\^P}{\^Q}$ by multiplying $\tilde{d}$ by the estimate of the ratio.


The coarse estimator $\tilde{d}$ can be constructed via a \emph{coupling}.
For two distributions $\^P$ and $\^Q$ over $\Omega$, a coupling is a pair of joint random variables $(X,Y)$ such that $X\sim \^P$ and $Y\sim \^Q$.
It is well known that $\DTV{\^P}{\^Q} =  \Pr[\+O]{X'\neq Y'}$, where $\+O=(X',Y')$ is the \emph{optimal} coupling between $\^P$ and $\^Q$. Suppose we can construct a coupling $\+C=(X,Y)$ that is not far from the optimal coupling $\+O$. Formally, $\Pr[\+O]{X'\neq Y'}=\gamma \cdot \Pr[\+C]{X\neq Y}$ for some $\gamma=1/\text{poly}(n)$. 
Then $\tilde{d} = \Pr[\+C]{X\neq Y}$ is a coarse estimator of $\DTV{\^P}{\^Q}$. 
To estimate the ratio $\DTV{\^P}{\^Q} / \tilde{d}$, define a distribution $\pi$:
\begin{align}\label{eq:pi-definition-overview}
\omega \in \Omega, \quad  \pi(\omega)=\Pr[\+C]{X=\omega\mid X\neq Y}
\end{align} 
and define function $f$ such that
$$ \omega \in \Omega, \quad  f(\omega) = 
    \frac{\max\{0, \^P(\omega) - \^Q(\omega)\}}{\Pr[\+C]{X = \omega \land X\neq Y}}.
$$
We can verify that $\E[\pi]{f} = \frac{\DTV{\^P}{\^Q}}{\Pr[\+C]{X\neq Y}} = \frac{\DTV{\^P}{\^Q}}{\tilde{d}} =\gamma \geq \frac{1}{\mathrm{poly}(n)}$. Using the approach described above, if we can perform the following operations in polynomial time
\begin{enumerate}[label=(\alph*)]
    \item sample $\omega \sim \pi$;\label{item:sample-omega}
    \item compute $f(\omega)$ for any $\omega$; \label{item:compute-f-omega}
    \item compute $\Pr[\+C]{X\neq Y}$;\label{item:compute-pr-c}
\end{enumerate}
 then we can get a relative-error approximation of $\DTV{\^P}{\^Q}$. Specifically, we draw $T = \mathrm{poly}(\frac{1}{\gamma}) = \mathrm{poly}(n)$ independent samples $\omega_1,\ldots,\omega_T\sim \pi$ and compute $\bar{f}=\frac{1}{T}\sum_{i=1}^{T}f(\omega_i)$. By Chebyshev's inequality, with high probability, $\bar{f}\cdot \Pr[\+C]{X\neq Y}$ is an estimate of $\DTV{\^P}{\^Q}$ within relative error.

\subsection{Recursive coupling for mixtures of product distributions}\label{subsection:recursive-coupling}
It remains to construct a coupling that supports the three operations above. For product distributions, \cite{FGJW23} used a coordinate-wise greedy coupling $\+C=(X,Y)$, which couples each coordinate optimally and independently, and showed that $1/n \leq \E[\pi]{f} \leq 1$.
This approach does not directly extend to mixtures of product distributions, where all the coordinates are \emph{correlated}, so one cannot optimally couple all coordinates simultaneously.

Our main technical contribution in this part is a new \emph{explicit recursive coupling} for mixtures of product distributions, together with an algorithmic implementation of the three operations above. At a high level, the recursive structure draws inspiration from~\cite{KwonECM23}, who introduced a recursive argument to bound the TV-distance between two mixtures of product distributions. Their analysis, however, is non-constructive: it yields a \emph{fixed-error} bound on the TV-distance, but does not by itself provide the algorithmic interface required here. To obtain a relative-error approximation with \emph{arbitrary} $\epsilon$, we need a coupling whose failure probability can be computed, and whose failure-conditioned distribution can be sampled and evaluated. We design such a coupling tailored to mixtures of product distributions, adapt the recursive analysis of \cite{KwonECM23} to this explicit construction, and show that the three operations in \Cref{item:sample-omega}, \Cref{item:compute-f-omega}, and \Cref{item:compute-pr-c} can all be implemented in polynomial time.

We construct a coupling between two mixtures of product distributions, $\^P = \sum_{s=1}^{k_1} \alpha^{(s)} \bigotimes_{i=1}^n \^P^{(s)}_i$ and $\^Q = \sum_{t=1}^{k_2} \beta^{(t)} \bigotimes_{i=1}^n \^Q^{(t)}_i$.
Although coupling the first coordinate would bias the marginal distribution of the remaining coordinates, we can still extract the common parts of $\^P^{(s)}_1$ and $\^Q^{(t)}_1$, and reduce the problem to coupling \emph{another} mixture of product distributions on the remaining coordinates.
For each $c\in [q]$, we first compute a lower bound for each active mixture component:
\begin{align}\label{eq:ell-definition}
    \ell\left(c\right) \triangleq \min\left\{ \min_{s\in [k_1]: {\alpha}^{(s)}>0} \^P^{(s)}_1(c),  \min_{t\in [k_2]: \beta^{(t)}>0} \^Q^{(t)}_1(c)\right\},
\end{align}
where $\^P^{(s)}_1$ (resp. $\^Q^{(t)}_1$) is the marginal distribution of the first coordinate of the $s$-th component of $\^P^{(s)}$ (resp. $\^Q^{(t)}$).
For $\omega\in [q]^n$ with $\omega_1=c$, we can write the probability of $\omega$ as
\begin{align*}
    \^P(\omega) &= \ell\left(c\right) \sum_{s=1}^{k_1}  \alpha^{(s)} \prod_{i=2}^n \^P^{(s)}_i(\omega_i) + {\sum_{s=1}^{k_1} \alpha^{(s)} \left(\^P^{(s)}_1(c) - \ell\left(c\right)\right) \prod_{i=2}^n \^P^{(s)}_i(\omega_i)}\\
 \tp{\alpha_c^{(s)} \defeq \frac{\alpha^{(s)} \left(\^P^{(s)}_1(c) - \ell\left(c\right)\right)}{\left(\^P_1(c)-\ell(c)\right)}} \quad   &=\ell\left(c\right) \underbrace{\sum_{s=1}^{k_1}  \alpha^{(s)} \prod_{i=2}^n \^P^{(s)}_i(\omega_i)}_{A_{\^P}} + \left(\^P_1(c)-\ell(c)\right) \underbrace{\sum_{s=1}^{k_1} \alpha_c^{(s)}  \prod_{i=2}^n \^P^{(s)}_i(\omega_i)}_{B_{\^P,c}},
\end{align*}
where $\alpha_c^{(s)}$ is a valid distribution over $s \in [k_1]$. Note that $\^P_1$ in the denominator is the marginal distribution of the mixture $\^P$ over the first coordinate.
Similarly, let $\beta_c^{(t)} \defeq \frac{\beta^{(t)} \left(\^Q^{(t)}_1(c) - \ell\left(c\right)\right)}{\left(\^Q_1(c)-\ell(c)\right)}$ and
\begin{align*}
    \^Q(\omega) = \ell\left(c\right) \underbrace{\sum_{t=1}^{k_2}  \beta^{(t)} \prod_{i=2}^n \^Q^{(t)}_i(\omega_i)}_{A_{\^Q}} + \left(\^Q_1(c)-\ell(c)\right) \underbrace{\sum_{t=1}^{k_2}  \beta_c^{(t)}  \prod_{i=2}^n \^Q^{(t)}_i(\omega_i)}_{B_{\^Q,c}}.
\end{align*}
Each of $A_{\^P}, A_{\^Q}$, $B_{\^P,c}$, and $B_{\^Q,c}$ is a mixture of product distributions over the remaining coordinates.
The component distributions in $A_{\^P}$ and $B_{\^P,c}$ are obtained from the components of $\^P$ by removing the first coordinate; the same relation holds between $A_{\^Q}, B_{\^Q,c}$ and $\^Q$.
However, the mixing weights may be different from those in $\^P$ and $\^Q$.

We use the above decomposition to construct a coupling $(X,Y)$ between $\^P$ and $\^Q$.
We first couple the first coordinate of $X$ and $Y$. For each $c\in [q]$, with probability $\ell(c)$, we set $X_1=Y_1=c$ and then couple $A_{\^P}$ and $A_{\^Q}$ \emph{recursively} for the remaining coordinates. 
For each $c \in [q]$,
with probability $\min\{\^P_1(c)-\ell(c), \^Q_1(c)-\ell(c)\}$, we set $X_1=Y_1=c$ and couple $B_{\^P,c}$ and $B_{\^Q,c}$ \emph{recursively} for the remaining coordinates. 
For the remaining probability mass\footnote{Specifically, the remaining probability for $X_1 = c$ is $\max\{0, \^P_1(c)-\^Q_1(c)\}$ and the remaining probability for $Y_1 = c$ is $\max\{0, \^Q_1(c)-\^P_1(c)\}$.} of the first coordinate, there is no common part with $X_1=Y_1=c$ for any $c\in [q]$. We arbitrarily couple this remaining mass; it must hold that $X_1\neq Y_1$, and hence $X\neq Y$. We call this case the \emph{coupling failure} case.
Suppose that after coupling the first coordinate, we have $X_1 = c \neq c' = Y_1$ for some $c,c'\in [q]$. We then independently sample the remaining coordinates according to $B_{\^P,c}$ and $B_{\^Q,c'}$, respectively.

The above recursive process terminates either when coupling fails at some coordinate or when it reaches the last coordinate. Denote the resulting recursive coupling by $\+C_{\-{RC}}$. Adapting the recursive TV-distance analysis of~\cite{KwonECM23} to this explicit coupling, we show that $\+C_{\-{RC}}$ is not far from the optimal coupling:
\begin{align*}
\Pr[\+C_{\-{RC}}]{X\neq Y} \leq (4nq)^{k_1+k_2-1} \cdot \DTV{\^P}{\^Q}.
\end{align*}
Therefore, we set $\gamma=(4nq)^{-k_1-k_2+1}$, which is inverse-polynomial in $n$ and $q$ when $k_1$ and $k_2$ are constants. 

\subsection{Implementation of three operations}
We now give polynomial-time algorithms for the three operations in \Cref{item:sample-omega}, \Cref{item:compute-f-omega}, and \Cref{item:compute-pr-c}. 
We create a state space $\+S$ to keep track of the recursive coupling process. A state $S\in \+S$ is a tuple $S=(i,\bar \alpha,\bar \beta)$, where $i\in [n]$ means we need to couple the remaining coordinates $i,i+1,\ldots, n$, and $\bar \alpha = (\bar \alpha^{(s)})_{s\in [k_1]}$ and $\bar \beta = (\bar \beta^{(t)})_{t\in [k_2]}$ are the \emph{current} mixing weights for each mixture component.
\Cref{subsection:recursive-coupling} describes the recursive coupling starting from the root state $(1,\alpha,\beta)$, and hence defines the quantities $A_{\^P}, A_{\^Q}, B_{\^P,c}, B_{\^Q,c}$ and $\ell(c)$ in that setting.
For the implementation, we use the same decomposition at an arbitrary state $S=(i,\bar \alpha,\bar \beta)$, replacing the first coordinate by the current coordinate $i$ and the original mixing weights by the current weights $\bar \alpha,\bar \beta$.
Thus, 
$$\ell(c)=\min \left\{\min_{s \in [k_1]: \bar \alpha^{(s)}>0} \^P_i^{(s)}(c),\min_{t \in [k_2]: \bar \beta^{(t)}>0} \^Q_i^{(t)}(c)\right\}.$$
Similarly, the updated weights $\bar \alpha_c=(\bar \alpha_c^{(s)})_{s\in [k_1]}$ in the $B_{\^P,c}$ branch are defined by
$$
\forall s \in [k_1],
\bar \alpha_c^{(s)}=\bar \alpha^{(s)} \cdot \frac{\^P_i^{(s)}(c)-\ell(c)}{\bar{\^P}_i(c)-\ell(c)} \text{ where } \bar{\^P}_i(c)=\sum_{s\in [k_1]}\bar \alpha^{(s)}\^P_i^{(s)}(c).$$
The $\bar \beta_c$ for the $B_{\^Q,c}$ is analogous.
There is also a failure state $\perp$. If we couple $X_i\neq Y_i$ (coupling failure case), the state $S$ transitions to the failure state $\perp$. 
If we couple $X_i=Y_i$ and then recursively couple $A_{\^P}$ and $A_{\^Q}$ for the remaining coordinates (we call it the good case A), the state $S$ transitions to $S'=(i+1,\bar \alpha, \bar \beta)$. If we couple $X_i=Y_i$ and then recursively couple $B_{\^P,c}$ and $B_{\^Q,c}$ for the remaining coordinates (we call it the good case B), the state $S$ transitions to $S'_c=(i+1,\bar \alpha_{c}, \bar \beta_{c})$.


The whole recursive coupling process is a random walk on the state space $\+S$ starting from the root state $S_{\mathrm{root}}=(1,\alpha, \beta)$. Every non-failure state $S = (i,\bar \alpha,\bar \beta) \neq \perp$ can transition to a state in $\{\perp\} \cup \{S'\} \cup \{S'_c\}_{c\in [q]}$, and it transitions to some $S'_c$ if and only if the good case B occurs.
The key observation is that case $B$ can occur at most $k_1+k_2-2$ times in the whole recursive coupling process.
By the definition of $\ell(c)$, $\ell(c)$ is equal to some $\^P_i^{(s)}(c)$ or $\^Q_i^{(t)}(c)$. Suppose $\ell(c) = \^P_i^{(s)}(c)$. Then, the new mixing weights in $S'_c$ satisfy
$\bar \alpha_{c}^{(s)} = \bar \alpha^{(s)} \cdot \frac{\^P_i^{(s)}(c) - \ell(c)}{\^P_i(c) - \ell(c)} = 0$.
Thus, whenever case $B$ occurs, the total number of positive mixing weights decreases by at least one. Since the root state has $k_1+k_2$ positive mixing weights, case $B$ can occur at most $k_1+k_2-2$ times.
Using this property, we can show that the total number of states is $|\+S| = (nq)^{O(k_1+k_2)}$, which is polynomial in $n$ and $q$.

Given a state $S$, it is easy to compute the transition probabilities to the next states. 
The whole transition graph forms a DAG with polynomial size.
To sample $\omega \sim \pi$ in~\eqref{eq:pi-definition-overview}, it suffices to sample a trajectory starting from the root state, conditioned on the random walk reaching the failure state $\perp$. 
We first run a dynamic programming algorithm to compute the probability of reaching $\perp$ from each state in $\+S$. Then, we use these probabilities to sample the desired conditional trajectory.
Similarly, using dynamic programming, we can compute $f(\omega)$ and $\Pr[\+C_{\-{RC}}]{X\neq Y}$ in polynomial time.



\begin{remark}[Comparison with the coupling in~\cite{BGMMPV24ICML}]
Finally, it is worth comparing our coupling with the coupling used in~\cite{BGMMPV24ICML}. They proposed a simpler coupling and applied it to approximating the TV-distance between two Bayesian networks.
Given two $n$-dimensional distributions $\^P$ and $\^Q$ over $[q]^n$, they couple the coordinates one by one using the optimal coupling of the conditional marginals. Suppose the first $i-1$ coordinates of $X$ and $Y$ have been generated and coupled consistently such that $X_{1:i-1}=Y_{1:i-1}= \sigma \in [q]^{i-1}$. For the $i$-th coordinates $X_i$ and $Y_i$, they use the optimal coupling of the conditional marginals $\^P_i$ and $\^Q_i$ given $\sigma$. The coupling terminates if $X_i \neq Y_i$ at some coordinate or when it reaches the last coordinate.

Their coupling is \emph{different} from ours. 
Unlike Bayesian networks, for general mixtures of product distributions, the conditional marginals at the $i$-th coordinate require the \emph{full information} of the prefix condition $\sigma \in [q]^{i-1}$. Since there are $O(q^n)$ possible prefix conditions $\sigma$, this coupling needs to consider exponentially many possible distributions for our problem.
In our coupling, even when the good event $X_i = Y_i$ occurs, we further split it into two sub-cases, $A$ and $B$. This splitting allows us to exploit the structure of mixtures of product distributions, so the total number of possible distributions appearing in our coupling is bounded by a polynomial.
\end{remark}

\section{A General Approach for Approximating the TV-Distance}\label{section:general-approximate-dtv}

When $k_1 = k_2 = 1$, the problem reduces to approximating the TV-distance between two product distributions $P$ and $Q$ over $\Omega = [q]^n$. This special case was studied in~\cite{FGJW23}. Although the algorithm there is tailored to product distributions, its underlying idea extends to a more general framework for approximating TV-distances for general discrete distributions.

Let $\^P$ and $\^Q$ be arbitrary discrete distributions over $\Omega$. A \emph{coupling} of $\^P$ and $\^Q$ is a pair of jointly distributed random variables $(X, Y)$ such that $X \sim \^P$ and $Y \sim \^Q$. For every coupling $\+C$ of $\^P$ and $\^Q$, the coupling inequality states that
\begin{align*}
    \DTV{\^P}{\^Q}  \leq \Pr[\+C]{X \neq Y}.
\end{align*}
A coupling $\+C$ is called \emph{optimal} if equality holds, namely if $\Pr[\+C]{X \neq Y} = \DTV{\^P}{\^Q}$. An optimal coupling exists for every pair of distributions $\^P$ and $\^Q$.

The following fact holds for every coupling $\+C$ of $\^P$ and $\^Q$.
\begin{fact}\label{fact:fact}
For any coupling $\+C$ of $\^P$ and $\^Q$, it holds that
\begin{align*}
    \forall \omega \in \Omega, \quad \Pr[\+C]{X = \omega \land X \neq Y} \geq \max\{0, \^P(\omega) - \^Q(\omega)\}.
    \end{align*}
If $\+C$ is an optimal coupling, then for all $\omega \in \Omega$, $\Pr[\+C]{X = \omega \land X \neq Y} = \max\{0, \^P(\omega) - \^Q(\omega)\}$.
\end{fact}
\begin{proof}
By definition,
\begin{align*}
    \Pr[\+C]{X=\omega\land X\neq Y}=\^P(\omega)-\Pr[\+C]{X=Y=\omega}
    \text{ and }
    \Pr[\+C]{X=Y=\omega}\leq\min\{\^P(\omega),\^Q(\omega)\}.
    \end{align*}
Therefore,
\begin{align*}
    \Pr[\+C]{X=\omega\land X\neq Y}\geq\max\{0,\^P(\omega)-\^Q(\omega)\}.
\end{align*}
When $\+C$ is optimal, we have 
\begin{align*}
    \DTV{\^P}{\^Q} = \Pr[\+C]{X\neq Y} = \sum_{\omega \in \Omega} \Pr[\+C]{X=\omega\land X\neq Y} \overset{(\ast)}{\geq} \sum_{\omega \in \Omega} \max\{0,\^P(\omega)-\^Q(\omega)\} = \DTV{\^P}{\^Q},
\end{align*}
which implies that the equality in $(\ast)$ must be achieved by the optimal coupling, i.e., for any $\omega \in \Omega$, $\Pr[\+C]{X=\omega\land X\neq Y}=\max\{0,\^P(\omega)-\^Q(\omega)\}$. The inequality $(\ast)$ itself is due to the fact that for any valid coupling, $\Pr{X = \omega} = \^P(\omega)$ and $\Pr{Y = \omega} = \^Q(\omega)$, and thus 
\begin{align*}
\Pr[\+C]{X = \omega \land X \neq Y} &= \Pr{X = \omega} - \Pr{X = \omega \land Y = \omega} \\
&\geq \^P(\omega) - \min\{\^P(\omega),\^Q(\omega)\} = \max\{0,\^P(\omega)-\^Q(\omega)\}. \qedhere
\end{align*}
\end{proof}

We now state an abstract version of the algorithm from~\cite{FGJW23}, formulated for two arbitrary discrete distributions $\^P$ and $\^Q$. Let $\gamma,T_0,T_1,T_2,T_3 > 0$ be parameters.

\begin{assumption}\label{assumption:probability-mass-oracle}
    There exists a probability mass oracle that, for any $\omega \in \Omega$, returns $\^P(\omega)$ and $\^Q(\omega)$ in time $T_0$. 
\end{assumption}

\begin{assumption}\label{assumption:coupling}
    There exists a coupling $\+C$ between $\^P$ and $\^Q$ such that $\frac{\DTV{\^P}{\^Q}}{\Pr[\+C]{X \neq Y}} \geq \gamma$, together with an oracle that can be preprocessed in time $T_1$ and supports the following query:
    \begin{itemize}
        \item Discrepancy query: return the value of $\Pr[\+C]{X \neq Y}$ in time $O(1)$.
    \end{itemize}
    Furthermore, if $\Pr[\+C]{X \neq Y} > 0$, then the oracle also supports the following two queries:
    \begin{itemize}
        \item Sampling query: draw an independent sample of $X$ from the conditional distribution of $X$ given $X \neq Y$ under $\+C$ in time $T_2$. Equivalently, it returns $\omega \in \Omega$ with probability $\Pr[\+C]{X = \omega \mid X \neq Y}$.
        \item Evaluation query: given any $\omega \in \Omega$, return the value of $\Pr[\+C]{X = \omega \land X \neq Y}$ in time $T_3$.
    \end{itemize}
\end{assumption}

Under \Cref{assumption:probability-mass-oracle} and \Cref{assumption:coupling}, the following theorem gives an algorithm for approximating the TV-distance between two \emph{arbitrary} discrete distributions $\^P$ and $\^Q$.

\begin{theorem}[\text{\cite{FGJW23}}]\label{theorem:general-approach}
Suppose \Cref{assumption:probability-mass-oracle} and   \Cref{assumption:coupling} hold for distributions $\^P$ and $\^Q$ with parameters $\gamma,T_0,T_1,T_2,T_3$. Then there exists a randomized algorithm such that, given any $\epsilon > 0$, the parameters $\gamma$, and access to the two oracles above, it outputs a random number $\hat{d}$ in time $O(T_1 + \frac{T_0+T_2+T_3}{\gamma \varepsilon^2})$ satisfying
$
\Pr{(1-\epsilon)\DTV{\^P}{\^Q} \leq \hat{d} \leq (1+\epsilon)\DTV{\^P}{\^Q}} \geq 99\%.
$
\end{theorem}

\begin{proof}
If $\Pr[\+C]{X\neq Y} = 0$, then $0 \leq \DTV{\^P}{\^Q} \leq \Pr[\+C]{X\neq Y}$ implies $\DTV{\^P}{\^Q}=0$, so we may directly return $\hat d = 0$. We therefore assume that $\Pr[\+C]{X\neq Y} > 0$. By \Cref{assumption:coupling}, after preprocessing time $T_1$ we can obtain the value of $\Pr[\+C]{X\neq Y}$ in constant time. Let $\pi$ denote the conditional distribution of $X$ given the event $X\neq Y$ under the coupling $\+C$, namely,
\begin{align*}
\forall \omega \in \Omega, \quad \pi(\omega) = \Pr[\+C]{X=\omega\mid X\neq Y} = \frac{\Pr[\+C]{X=\omega \land X\neq Y}}{\Pr[\+C]{X\neq Y}}.
\end{align*}
For all $\omega \in \Omega$ with $\pi(\Omega) \neq 0$, it holds that $\Pr[\+C]{X = \omega \land X\neq Y} > 0$, and we define 
$$
f(\omega) = \frac{\max\{0, \^P(\omega) - \^Q(\omega)\}}{\Pr[\+C]{X = \omega \land X\neq Y}}.
$$
By \Cref{fact:fact}, $0 \leq f(\omega) \leq 1$ for all $\omega \in \Omega$ with $\pi(\omega) > 0$, hence $\E[\pi]{f^2} \leq \E[\pi]{f}$, and $$\Var[\pi]{f} = \E[\pi]{f^2} - (\E[\pi]{f})^2 \leq \E[\pi]{f^2} \leq \E[\pi]{f}.$$ 
Now let $\Omega_+=\{\omega\in \Omega \mid \pi(\omega)>0\}$. Then
\begin{align*}
\E[\pi]{f} = \sum_{\omega \in \Omega_+}\pi(\omega)\frac{\max\{0, \^P(\omega) - \^Q(\omega)\}}{\Pr[\+C]{X = \omega \land X\neq Y}} = \sum_{\omega\in \Omega_+}\frac{\max\{0, \^P(\omega) - \^Q(\omega)\}}{\Pr[\+C]{X\neq Y}}.
\end{align*}
For $\omega \in \Omega$ with $\pi(\omega)=0$, then either $\Pr[\+C]{X = \omega} = 0$ (hence $\^P(\omega)=0$), or $\Pr[\+C]{X\neq Y \mid X = \omega} = 0$ (which implies that whenever $X=\omega$, we also have $Y=\omega$, and therefore $\^Q(\omega)\geq \^P(\omega)$). Hence,
\begin{align}\label{eq:f-omega-0}
 \pi(\omega) = 0 \quad \implies \quad \max\{0,\^P(\omega)-\^Q(\omega)\}=0.
\end{align}
By~\eqref{eq:f-omega-0}, it holds that $\max\{0, \^P(\omega) - \^Q(\omega)\} \neq 0$ implies $\pi(\omega) > 0$. Hence, $\Omega_+$ contains all $\omega \in \Omega$ such that $\max\{0, \^P(\omega) - \^Q(\omega)\} \neq 0$. Therefore,
\begin{align*}
\DTV{\^P}{\^Q} = \sum_{\omega \in \Omega}\max\{0, \^P(\omega) - \^Q(\omega)\}
= \sum_{\omega \in \Omega_+}\max\{0, \^P(\omega) - \^Q(\omega)\},
\end{align*}
we obtain
\begin{align*}
\E[\pi]{f} = \frac{\DTV{\^P}{\^Q}}{\Pr[\+C]{X\neq Y}}\in[\gamma, 1].
\end{align*}

By the sampling query in \Cref{assumption:coupling}, we can draw a sample from $\pi$ in time $T_2$. Moreover, using \Cref{assumption:probability-mass-oracle} and \Cref{assumption:coupling}, for any $\omega \in \Omega$ we can compute
\[
\^P(\omega), \quad \^Q(\omega), \quad \Pr[\+C]{X=\omega\land X\neq Y}, \quad f(\omega)
\]
in time $T_0 + T_3$. 
Let $m = \frac{100}{\gamma \varepsilon^2}$. We draw $m$ independent samples $\omega^{(1)}, \omega^{(2)}, \ldots \omega^{(m)} \sim \pi$ and compute $\bar{f} = \frac{1}{m}\sum_{i=1}^mf(\omega^{(i)})$. By Chebyshev's inequality,
\begin{align*} 
    \Pr[]{\left|\bar f - \E[\pi]{f}\right| \geq \varepsilon \E[\pi]{f}} \leq \frac{\Var[\pi]{f}}{m\varepsilon^2 \E[\pi]{f}^2} \leq \frac{1}{m\varepsilon^2 \E[\pi]{f}} \leq \frac{1}{\frac{100}{\gamma \varepsilon^2}\varepsilon^2 \gamma} = 0.01.
\end{align*}
Thus, with probability at least $99\%$, $\bar f \in \left[(1-\varepsilon)\E[\pi]{f}, (1+\varepsilon)\E[\pi]{f}\right]$. The proof is completed by outputting $\hat d = \bar{f} \cdot \Pr[\+C]{X\neq Y}$, where $\Pr[\+C]{X\neq Y}$ is the value returned by the discrepancy query in \Cref{assumption:coupling} in constant time.
\end{proof}

\section{Algorithm for Mixtures of Product Distributions}

To apply \Cref{theorem:general-approach} to the problem of approximating the TV-distance between two mixtures of product distributions, we need to verify  \Cref{assumption:probability-mass-oracle} and \Cref{assumption:coupling}. By the definition of mixture of product distributions in~\eqref{equation:mixture-of-product-distributions}, the first assumption holds with $T_0 = O(n(k_1 + k_2))$. Our main task is to construct a coupling to verify the second assumption.

\begin{lemma}\label{lemma:coupling-mix-prod}
    Let $k_1,k_2 \geq 1$ be two constants. For any two mixtures of product distributions $\^P = \sum_{s=1}^{k_1} \alpha^{(s)} \bigotimes_{i=1}^n \^P^{(s)}_i$ and $\^Q = \sum_{t=1}^{k_2} \beta^{(t)} \bigotimes_{i=1}^n \^Q^{(t)}_i$, there exists a coupling $\+C$ between $\^P$ and $\^Q$ such that \Cref{assumption:coupling} holds with parameters $\gamma = (4nq)^{-k_1 - k_2 + 1}$, $T_1,T_2,T_3 = O(q^2(nq)^{k_1 + k_2 - 1})$.
\end{lemma}

\Cref{thm:tv-mix-prod} is a simple corollary of \Cref{theorem:general-approach} and \Cref{lemma:coupling-mix-prod}.
The rest of this section is devoted to the proof of the lemma.

\subsection{The recursive coupling for mixtures of product distributions}

In this subsection, we construct a recursive coupling between two mixtures of product distributions $\^P$ and $\^Q$ in two steps. In \Cref{subsection:recursive-sampling-process}, we first introduce a recursive procedure for drawing a sample from a single mixture $\^P$. Building on this procedure, in \Cref{subsection:recursive-coupling-process} we then construct a recursive coupling between $\^P$ and $\^Q$.
Our construction is inspired by the technique of~\cite{KwonECM23}, who develop a recursive analysis that bounds the total variation distance between two mixtures of product distributions up to a \emph{fixed} error. 
Our main contribution is a new recursive coupling process, which is crucial for approximating the TV-distance with \emph{arbitrary} relative error.
In \Cref{subsection:implementation-oracles}, we use this process to implement an oracle that answers the queries in \Cref{assumption:coupling} within polynomial time.


\subsubsection{The recursive sampling process}\label{subsection:recursive-sampling-process}

We now describe a recursive sampling procedure that generates a sample $X$ from the mixture of product distributions
$$
\^P = \sum_{s=1}^{k_1} \alpha^{(s)} \bigotimes_{i=1}^n \^P^{(s)}_i .
$$
This sampling procedure will later serve as the basis for our coupling construction.
Throughout the recursion, the component product distributions $\^P^{(s)}$ remain fixed, while the mixing weights are updated from one recursive step to the next. We use $(\bar\alpha^{(s)})_{s=1}^{k_1}$ to denote the mixing weights at the current recursive step, to distinguish them from the original mixing weights $(\alpha^{(s)})_{s=1}^{k_1}$.
At recursion depth $j$, let
\[
\bar{\^P} = \sum_{s=1}^{k_1} \bar{\alpha}^{(s)} \bigotimes_{i=j}^{n}\^P^{(s)}_i
\]
be the current mixture over $[q]^{n-j+1}$. We consider the recursive coordinate-wise procedure
\[
\text{RecursiveSample}\left(j, \bar \alpha\right),
\]
where $j \in [n]$ is the current coordinate. At this point, the coordinates $X_1, \ldots, X_{j-1}$ have already been fixed, and the task is to generate $X_j$ from the current mixture $\bar{\^P}$ and then update the mixing weights for the next recursive step. To describe the procedure, we define a lower-bound selector.
\begin{definition}[lower-bound selector]\label{def:lower-bound selector}
    A lower-bound selector $\ell^{\^P}$ is a function $\ell^{\^P}\left(j, \bar\alpha, c\right)$ defined for any coordinate $j \in [n]$, any mixing-weight vector $\bar \alpha = (\bar{\alpha}^{(s)})_{s=1}^{k_1}$, and any $c \in [q]$, satisfying: 
    \begin{align}\label{eq:probability-lower-bound-inequality}
        \min_{s \in [k_1]: \bar{\alpha}^{(s)} > 0}\^P_j^{(s)}(c)\geq \ell^{\^P}(j, \bar{\alpha}, c) \geq 0.
    \end{align}
\end{definition}

Here the minimum ranges over all components $s \in [k_1]$ with $\bar{\alpha}^{(s)} > 0$, namely over the components that are still active at the current recursive step. Thus $\ell^{\^P}(j,\bar\alpha,c)$ is required to be a nonnegative lower bound on the marginal probability of seeing value $c$ in coordinate $j$ across all active components.
The above definition does \emph{not} specify a unique lower-bound selector $\ell^{\^P}$. Any function satisfying~\eqref{eq:probability-lower-bound-inequality} is valid.
Let us fix an arbitrary lower-bound selector $\ell^{\^P}$ satisfying~\eqref{eq:probability-lower-bound-inequality}.
Our recursive sampling process in \Cref{subsection:recursive-sampling-process} works for any lower-bound selector $\ell^{\^P}$. 
In \Cref{subsection:recursive-coupling-process}, when constructing the coupling, we will specify a particular choice of $\ell^{\^P}$.


For the current recursive state $(j, (\bar{\alpha}^{(s)})_{s=1}^{k_1})$, we first compute the lower bounds $\ell^{\^P}\left(j, \bar\alpha, c\right)$ for all $c\in [q]$.
Fix a value $c \in [q]$ and an assignment $\omega \in [q]^{n-j+1}$ satisfying $\omega_j = c$.
The probability of $\omega$ under the current mixture $\bar{\^P}$ can be written as
\begin{align*}
\bar{\^P}(\omega) = \sum_{s=1}^{k_1} \bar \alpha^{(s)} \prod_{i=j}^{n} \^P^{(s)}_i(\omega_i) = \sum_{s=1}^{k_1} \bar\alpha^{(s)} \^P^{(s)}_j(c) \prod_{i=j+1}^n \^P^{(s)}_i(\omega_i),
\end{align*}
where the second equality uses the assumption $\omega_j = c$.
By~\eqref{eq:probability-lower-bound-inequality}, it holds that $\^P^{(s)}_j(c) \geq \ell^{\^P}\left(j, \bar\alpha, c\right)$ for all $s \in [k_1]$ with $\bar \alpha^{(s)} > 0$. Hence, for each $s$ with $\bar \alpha^{(s)} > 0$, we can decompose $\^P^{(s)}_j(c)$ as two non-negative terms: $\ell^{\^P}\left(j, \bar\alpha, c\right)$ and $\^P^{(s)}_j(c) - \ell^{\^P}\left(j, \bar\alpha, c\right)$. We have the following decomposition:
\begin{align}\label{eq:P-decomposition-1}
    \bar{\^P}(\omega) = \ell^{\^P}\left(j, \bar\alpha, c\right) \underbrace{\sum_{s=1}^{k_1} \bar \alpha^{(s)} \prod_{i=j+1}^n \^P^{(s)}_i(\omega_i)}_{A} + {\sum_{s=1}^{k_1} \bar \alpha^{(s)} \left(\^P^{(s)}_j(c) - \ell^{\^P}\left(j, \bar\alpha, c\right)\right) \prod_{i=j+1}^n \^P^{(s)}_i(\omega_i)}.
\end{align}
The above summation enumerates all $s \in [k_1]$, which may contain inactive components $s$ with $\bar \alpha^{(s)} = 0$. However, those inactive components only contribute 0 to the whole sum.
In~\eqref{eq:P-decomposition-1},
the first term is the product of $\ell^{\^P}\left(j, \bar\alpha, c\right)$ and $A$, where $A$ is exactly the probability of the suffix $(\omega_i)_{i=j+1}^n$ under the current mixture $\sum_{s=1}^{k_1} \bar\alpha^{(s)} \bigotimes_{i=j+1}^n \^P^{(s)}_i$. 
To write the second term in a similar form, we introduce a modified mixing weight $\bar\alpha_{j,c}^{(s)}$ for all $s \in [k_1]$ as follows:
\begin{align*}
    \bar\alpha_{j,c}^{(s)} \defeq  
    \begin{cases}
        \bar{\alpha}^{(s)}, &\bar{\^P}_j(c) - \ell^{\^P}\left(j, \bar\alpha, c\right) = 0;\\
        \bar\alpha^{(s)} \frac{\^P^{(s)}_j(c) - \ell^{\^P}\left(j, \bar\alpha, c\right)}{\bar{\^P}_j(c) - \ell^{\^P}\left(j, \bar\alpha, c\right)}, & \text{otherwise,}
    \end{cases} \quad \text{where } \bar{\^P}_j(c) = \sum_{s=1}^{k_1} \bar\alpha^{(s)} \^P^{(s)}_j(c).
\end{align*}
We next verify that $\bar\alpha_{j,c}$ defines a valid mixing-weight vector.
By \Cref{def:lower-bound selector}, $\^P^{(s)}_j(c) - \ell^{\^P}\left(j, \bar\alpha, c\right) \geq 0$, and $\bar{\^P}_j(c) = \sum_{s=1}^{k_1} \bar\alpha^{(s)} \^P^{(s)}_j(c)$ is an average of the active values $\^P^{(s)}_j(c)$, therefore $\bar{\^P}_j(c) \geq \ell^{\^P}\left(j, \bar\alpha, c\right)$.
If $\bar{\^P}_j(c) - \ell^{\^P}\left(j, \bar\alpha, c\right) = 0$, $\bar{\alpha}_{j,c} = \bar{\alpha}$ is a valid mixing-weight vector as long as $\bar{\alpha}$ is valid. 
Otherwise, when $\bar{\^P}_j(c) > \ell^{\^P}\left(j, \bar\alpha, c\right)$,
\begin{align*}
    \sum_{s=1}^{k_1} \bar\alpha_{j,c}^{(s)} = \sum_{s=1}^{k_1} \bar\alpha^{(s)} \frac{\^P^{(s)}_j(c) - \ell^{\^P}\left(j, \bar\alpha, c\right)}{\bar{\^P}_j(c) - \ell^{\^P}\left(j, \bar\alpha, c\right)} = \frac{\sum_{s=1}^{k_1} \bar\alpha^{(s)}\^P^{(s)}_j(c) }{\bar{\^P}_j(c) - \ell^{\^P}\left(j, \bar\alpha, c\right)} - \frac{\ell^{\^P}\left(j, \bar\alpha, c\right)}{\bar{\^P}_j(c) - \ell^{\^P}\left(j, \bar\alpha, c\right)} = 1.
\end{align*}
Hence, for $\omega \in [q]^{n-j+1}$ with $\omega_j = c$, we can rewrite $\bar{\^P}(\omega)$ as
\begin{align}\label{eq:P-decomposition}
    \bar{\^P}(\omega)
    = \ell^{\^P}\left(j, \bar\alpha, c\right)\sum_{s=1}^{k_1} \bar\alpha^{(s)} \prod_{i=j+1}^n \^P^{(s)}_i(\omega_i)
    + \left(\bar{\^P}_j(c) - \ell^{\^P}\left(j, \bar\alpha, c\right)\right)\sum_{s=1}^{k_1} \bar\alpha_{j,c}^{(s)} \prod_{i=j+1}^n \^P^{(s)}_i(\omega_i).
\end{align}

Equation~\eqref{eq:P-decomposition} naturally suggests a recursive sampling rule. Given the state $(j, (\bar{\alpha}^{(s)})_{s=1}^{k_1})$, we first sample the current coordinate $X_j$, and then reduce the task of sampling the suffix $X_{j+1:n}$ to sampling from one of two smaller mixtures. Concretely, once the current coordinate is set to $c$, the suffix distribution is either
$\sum_{s=1}^{k_1} \bar\alpha^{(s)} \bigotimes_{i=j+1}^n \^P^{(s)}_i$
or $\sum_{s=1}^{k_1} \bar\alpha_{j,c}^{(s)} \bigotimes_{i=j+1}^n \^P^{(s)}_i$.
The component distributions remain the same (except that their dimension is reduced by one), while the mixing weights may change from $\bar\alpha$ to $\bar\alpha_{j,c}$. Applying this rule recursively yields the sampling procedure shown in Algorithm \ref{alg:recursive-sampling}.


\begin{algorithm}[ht]
    \caption{Recursive sampling process for mixtures of product distributions}
    \label{alg:recursive-sampling}
    \SetKwInOut{Input}{Static Input}
    \SetKwInOut{Output}{Output}
    \SetKwInOut{Parameter}{Parameter}
    \underline{RecursiveSample} $(j, \bar\alpha)$\;
    \Parameter{an index $j \in [n+1]$ and mixing weights $\bar{\alpha}^{(1)},\bar{\alpha}^{(2)},\ldots,\bar{\alpha}^{(k_1)}$}
    \Input{product distributions $\^P^{(1)},\^P^{(2)},\ldots,\^P^{(k_1)}$ over $[q]^n$; an arbitrary lower-bound selector function $\ell^{\^P}$ satisfying~\eqref{eq:probability-lower-bound-inequality}}
    \Output{a random sample $X = (X_j,X_{j+1},\ldots,X_n) \sim \bar{\^P}= \sum_{s=1}^{k_1} \bar{\alpha}^{(s)} \bigotimes_{i=j}^n \^P^{(s)}_i$}
    \textbf{if}  $j > n$ \textbf{then} \Return $\emptyset$\;
    Compute the probability lower bound $\ell^{\^P}\left(j, \bar\alpha, c\right)$ for all $c\in [q]$\; 
    Sample a random pair $(c,e) \in [q] \times \{0,1\}$ such that for each $c \in [q]$, the probability of $(c,0)$ is $\ell^{\^P}\left(j, \bar\alpha, c\right)$ and the probability of $(c,1)$ is $\bar{\^P}_j(c) - \ell^{\^P}\left(j, \bar\alpha, c\right)$, where $\bar{\^P}_j(c) = \sum_{s=1}^{k_1} \bar{\alpha}^{(s)} \^P^{(s)}_j(c)$ is the marginal probability of coordinate $j$ under $\bar{\^P}$\label{line:sample-random-pair}\;
    \If{$e = 0$}{
        Sample $(X_{j+1},X_{j+2},\ldots,X_n) \gets \text{RecursiveSample}(j+1, \bar\alpha)$\;
    } \Else{
        Compute mixing weights $\bar{\alpha}_{j,c}^{(s)} = \bar{\alpha}^{(s)} \frac{\^P^{(s)}_j(c) - \ell^{\^P}(j,\bar\alpha, c)}{\bar{\^P}_j(c) - \ell^{\^P}(j,\bar\alpha,c)}$ for all $s \in [k_1]$\;
        Sample $(X_{j+1},X_{j+2},\ldots,X_n) \gets \text{RecursiveSample}(j+1, \bar{\alpha}_{j,c})$\;
    }
    \Return $\tp{X_j = c, X_{j+1},X_{j+2},\ldots,X_n}$\;
\end{algorithm}

In Algorithm \ref{alg:recursive-sampling}, the lower-bound selector function $\ell^{\^P}$ is treated as part of the input. The following lemma shows that the algorithm outputs a correct sample from $\^P$ as long as $\ell^{\^P}$ satisfies~\eqref{eq:probability-lower-bound-inequality}.

\begin{lemma}\label{lemma:validity-of-recursive-sampling}
Given a mixture of product distributions $\^P = \sum_{s=1}^{k_1} \alpha^{(s)} \bigotimes_{i=1}^n \^P^{(s)}_i$ and an arbitrary lower-bound selector function $\ell^{\^P}$ satisfying~\eqref{eq:probability-lower-bound-inequality}, the procedure \text{RecursiveSample}($1, \alpha$) in Algorithm \ref{alg:recursive-sampling} outputs a sample $X \sim \^P$.
\end{lemma}

\Cref{lemma:validity-of-recursive-sampling} follows by induction on the recursion depth using~\eqref{eq:P-decomposition}. The proof is in Appendix \ref{section:analysis-recursive-coupling}.

The same recursive procedure applies to any mixture of product distributions. In particular, we may use it to sample from
$\bar{\^Q}=\sum_{t=1}^{k_2} \bar\beta^{(t)} \bigotimes_{i=j}^{n} \^Q^{(t)}_i$
by replacing $k_1$ with $k_2$, replacing $(\bar\alpha^{(s)})_{s = 1}^{k_1}$ with $(\bar\beta^{(t)})_{t = 1}^{k_2}$, replacing $\^P^{(s)}$ with $\^Q^{(t)}$, and using an arbitrary lower-bound selector $\ell^{\^Q}(j, \bar\beta, c)$ satisfying
\begin{align}\label{eq:probability-lower-bound-inequality-for-Q}
 \min_{t\in[k_2]: \bar{\beta}^{(t)} > 0}\^Q^{(t)}_j(c) \geq \ell^{\^Q}(j, \bar \beta, c) \geq 0.
\end{align}
\Cref{lemma:validity-of-recursive-sampling} also shows that by substituting $P$ with $Q$ and $\alpha$ with $\beta$, the procedure \text{RecursiveSample}($1, \beta$) outputs a sample $Y \sim \^Q$.


\subsubsection{The recursive coupling process}\label{subsection:recursive-coupling-process}

We now define a coupling between two mixtures of product distributions
\[
\^P=\sum_{s=1}^{k_1} \alpha^{(s)} \bigotimes_{i=1}^n \^P^{(s)}_i
\qquad\text{and}\qquad
\^Q=\sum_{t=1}^{k_2} \beta^{(t)} \bigotimes_{i=1}^n \^Q^{(t)}_i.
\]
At a high level, the coupling runs Algorithm \ref{alg:recursive-sampling} on both $\^P$ and $\^Q$ to generate samples $X \sim \^P$ and $Y \sim \^Q$, while sharing randomness between the two executions so that the resulting pair $(X,Y)$ is a coupling of $\^P$ and $\^Q$.


We use $\text{RecursiveSample}(1, \alpha)$ to generate a sample $X \sim \^P$ and $\text{RecursiveSample}(1, \beta)$ to generate a sample $Y \sim \^Q$, and couple the two recursive procedures at the same time.
Suppose that at some point the execution for $\^P$ is at recursion state $\text{RecursiveSample}(j, \bar\alpha)$ and the execution for $\^Q$ is at recursion state $\text{RecursiveSample}(j, \bar\beta)$, where $\bar{\alpha} = (\bar{\alpha}^{(s)})_{s=1}^{k_1}$ and $\bar\beta = (\bar{\beta}^{(t)})_{t=1}^{k_2}$ denote the current mixing weights and may differ from the original vectors $\alpha$ and $\beta$.
We define a shared lower-bound selector $\ell(j, \bar\alpha, \bar\beta, c)$ as follows. 
\begin{definition}\label{definition:probability-lower-bound}
For any $j\in [n]$, $c \in [q]$, and valid mixing-weight vectors $\bar\alpha = (\bar\alpha^{(s)})_{s=1}^{k_1}$ and $\bar\beta = (\bar{\beta}^{(t)})_{t=1}^{k_2}$, define
\begin{align*}
    \ell\left(j, \bar\alpha, \bar\beta, c\right) \triangleq \min\left\{ \min_{s\in [k_1]: \bar{\alpha}^{(s)}>0} \^P^{(s)}_j(c),  \min_{t\in [k_2]: \bar{\beta}^{(t)}>0} \^Q^{(t)}_j(c)\right\}.
\end{align*}
\end{definition}

The one-sided lower-bound selectors used by the two recursive sampling procedures are then defined by
$$
\ell^{\^P}(j, \bar\alpha, c) = \ell^{\^Q}(j, \bar\beta, c) = \ell(j, \bar\alpha, \bar\beta, c).
$$
By construction, $\ell^{\^P}(j, \bar\alpha, c)$ satisfies~\eqref{eq:probability-lower-bound-inequality} for $\^P$, and $\ell^{\^Q}(j, \bar\beta, c)$ satisfies~\eqref{eq:probability-lower-bound-inequality-for-Q} for $\^Q$.
In Line~\ref{line:sample-random-pair} of Algorithm \ref{alg:recursive-sampling}, each process draws a random pair $(c,e) \in [q] \times \{0,1\}$. Let $(c_{\bar{\^P}},e_{\bar{\^P}})$ and $(c_{\bar{\^Q}},e_{\bar{\^Q}})$ denote the pairs drawn by $\text{RecursiveSample}(j, \bar{\alpha})$ and $\text{RecursiveSample}(j, \bar{\beta})$, respectively. Since both procedures use the same lower bounds $\ell(j, \bar\alpha, \bar\beta, c)$, we define a coupling $\bar{\+C}$ between these two pairs as follows:
\begin{itemize}
    \item For each $c \in [q]$, $\Pr[\bar{\+C}]{c_{\bar{\^P}} = c_{\bar{\^Q}} = c  \land e_{\bar{\^P}} = e_{\bar{\^Q}} = 0} = \ell(j, \bar\alpha, \bar\beta, c)$.
    \item For each $c \in [q]$, $\Pr[\bar{\+C}]{c_{\bar{\^P}} = c_{\bar{\^Q}} = c  \land e_{\bar{\^P}} = e_{\bar{\^Q}} = 1} = \min\{\bar{\^P}_j(c), \bar{\^Q}_j(c)\}-\ell(j, \bar\alpha, \bar\beta, c)$, where
    $\bar{\^P} = \sum_{s=1}^{k_1} \bar{\alpha}^{(s)} \bigotimes_{i=j}^n \^P^{(s)}_i$
    and
    $\bar{\^Q} = \sum_{t=1}^{k_2} \bar{\beta}^{(t)} \bigotimes_{i=j}^n \^Q^{(t)}_i$
    are the mixtures of product distributions at the current recursive step.
    \item Given the above two points, for each $c \in A = \{x \in [q] \mid \bar{\^P}_j(x) > \bar{\^Q}_j(x)\}$,
    the random variable $(c_{\bar{\^P}},e_{\bar{\^P}})$ has a remaining probability $\bar{\^P}_j(c) - \bar{\^Q}_j(c)$ of taking value $(c, 1)$; and for each $c \in B = \{x \in [q] \mid \bar{\^P}_j(x) < \bar{\^Q}_j(x)\}$,
    the random variable $(c_{\bar{\^Q}},e_{\bar{\^Q}})$ has remaining probability $\bar{\^Q}_j(c) - \bar{\^P}_j(c)$ of taking value $(c, 1)$. The total remaining probabilities on the two sides are equal:
    $\sum_{c \in A} (\bar{\^P}_j(c) - \bar{\^Q}_j(c))  = \sum_{c \in B} (\bar{\^Q}_j(c) - \bar{\^P}_j(c))$.
    We therefore couple the remaining mass arbitrarily. Since $A$ and $B$ are disjoint, in this case we necessarily have
    \begin{align*}
        c_{\bar{\^P}} \neq c_{\bar{\^Q}} \quad \text{and} \quad e_{\bar{\^P}} = e_{\bar{\^Q}} = 1.
    \end{align*}
\end{itemize}

Using the coupling $\bar{\+C}$ between $(c_{\bar{\^P}},e_{\bar{\^P}})$ and $(c_{\bar{\^Q}},e_{\bar{\^Q}})$, we now couple two sampling processes $\text{RecursiveSample}(j, \bar{\alpha})$ and $\text{RecursiveSample}(j, \bar{\beta})$ as follows:
\begin{itemize}
    \item Sample $(c_{\bar{\^P}},e_{\bar{\^P}})$ and $(c_{\bar{\^Q}},e_{\bar{\^Q}})$ according to $\bar{\+C}$.
    \item If $c_{\bar{\^P}} = c_{\bar{\^Q}} = c$ for some $c \in [q]$ and $e_{\bar{\^P}} = e_{\bar{\^Q}} = 0$, then we set $X_j = Y_j = c$ and recursively couple the suffixes via $\text{RecursiveCouple}(j+1, \bar{\alpha}, \bar{\beta})$.
    \item If $c_{\bar{\^P}} = c_{\bar{\^Q}} = c$ for some $c \in [q]$ and $e_{\bar{\^P}} = e_{\bar{\^Q}} = 1$, then we set $X_j = Y_j = c$ and recursively couple the suffixes via $\text{RecursiveCouple}(j+1, \bar{\alpha}_{j,c}, \bar{\beta}_{j,c})$.
    \item In the remaining case, we have $X_j = c_{\bar{\^P}} \neq c_{\bar{\^Q}} = Y_j$ and $e_{\bar{\^P}} = e_{\bar{\^Q}} = 1$. Hence the coupling already fails at coordinate $j$, and we simply sample the suffixes independently using $\text{RecursiveSample}(j+1, \bar{\alpha}_{j,c_{\bar{\^P}}})$ and $\text{RecursiveSample}(j+1, \bar{\beta}_{j,c_{\bar{\^Q}}})$.
\end{itemize}

The pseudocode of the coupling process is given in Algorithm \ref{alg:coupling-process}.

\begin{algorithm}[ht]
    \caption{Recursive coupling process for mixtures of product distributions}
    \label{alg:coupling-process}
    \SetKwInOut{Input}{Static Input}
    \SetKwInOut{Output}{Output}
    \SetKwInOut{Parameter}{Parameter}
    \underline{RecursiveCouple} $(j, \bar{\alpha}, \bar{\beta})$\;
    \Parameter{an index $j \in [n+1]$ and mixing weights $\bar{\alpha} = (\bar\alpha^{(s)})_{s=1}^{k_1}$, $\bar{\beta} = (\bar\beta^{(t)})_{t=1}^{k_2}$}
    \Input{product distributions $\^P^{(1)},\^P^{(2)},\ldots,\^P^{(k_1)}$ and $\^Q^{(1)},\^Q^{(2)},\ldots,\^Q^{(k_2)}$ over $[q]^n$; a fixed shared probability lower-bound selector $\ell$ as in \Cref{definition:probability-lower-bound}}
    \Output{joint random variables $(X,Y)$ such that $X = (X_j,X_{j+1},\ldots,X_n) \sim \bar{\^P}= \sum_{s=1}^{k_1} \bar{\alpha}^{(s)} \bigotimes_{i=j}^n \^P^{(s)}_i$ and $Y = (Y_j,Y_{j+1},\ldots,Y_n) \sim \bar{\^Q}= \sum_{t=1}^{k_2} \bar{\beta}^{(t)} \bigotimes_{i=j}^n \^Q^{(t)}_i$}
    \textbf{if}  $j > n$ \textbf{then} \Return $(\emptyset,\emptyset)$\;
    Compute $\ell(j, \bar\alpha, \bar\beta, c)$ for all $c\in [q]$, and then use the coupling $\bar{\+C}$ to jointly sample $(c_{\bar{\^P}},e_{\bar{\^P}})$ and $(c_{\bar{\^Q}},e_{\bar{\^Q}})$\label{line:couple-random-pair}\;
    \If{$ e_{\bar{\^P}} = e_{\bar{\^Q}} = 0$}{
        Sample $(X_{j+1:n},Y_{j+1:n}) \gets \text{RecursiveCouple}(j+1, \bar{\alpha}, \bar{\beta})$\label{line:recursive-couple-0}\;
    } \Else{
        Let $\bar{\alpha}_{j,c_{\bar{\^P}}}^{(s)} = \bar{\alpha}^{(s)} \frac{\^P^{(s)}_j(c_{\bar{\^P}}) - \ell(j, \bar\alpha, \bar\beta, c_{\bar{\^P}})}{\bar{\^P}_j(c_{\bar{\^P}}) - \ell(j, \bar\alpha, \bar\beta, c_{\bar{\^P}})}$ for $s \in [k_1]$ and $\bar{\beta}_{j,c_{\bar{\^Q}}}^{(t)} = \bar{\beta}^{(t)} \frac{\^Q^{(t)}_j(c_{\bar{\^Q}}) - \ell(j, \bar\alpha, \bar\beta, c_{\bar{\^Q}})}{\bar{\^Q}_j(c_{\bar{\^Q}}) - \ell(j, \bar\alpha, \bar\beta, c_{\bar{\^Q}})}$ for $t \in [k_2]$\;
        \If{$c_{\bar{\^P}} = c_{\bar{\^Q}}$}{
            Sample $(X_{j+1:n},Y_{j+1:n}) \gets \text{RecursiveCouple}(j+1, \bar{\alpha}_{j, c}, \bar{\beta}_{j,c})$\label{line:recursive-couple-1}, where $c = c_{\bar{\^P}} = c_{\bar{\^Q}}$\;
        } \Else{
            Sample $X_{j+1:n} \gets \text{RecursiveSample}(j+1, \bar{\alpha}_{j,c_{\bar{\^P}}})$ with $\ell^{\^P} = 0$\label{line:recursive-sample-1}\;
            Sample $Y_{j+1:n} \gets \text{RecursiveSample}(j+1, \bar{\beta}_{j,c_{\bar{\^Q}}})$ with $\ell^{\^Q} = 0$\label{line:recursive-sample-2}\;
        }
    }
    Let $X = (X_j =c_{\bar{\^P}}, X_{j+1:n})$ and $Y = (Y_j =c_{\bar{\^Q}}, Y_{j+1:n})$ and return $(X,Y)$\;
\end{algorithm}

\begin{lemma}[Validity of recursive coupling]\label{lemma:validity-coupling-mix-prod}
    Given two mixtures of product distributions $\^P = \sum_{s=1}^{k_1} \alpha^{(s)} \bigotimes_{i=1}^n \^P^{(s)}_i$ and $\^Q = \sum_{t=1}^{k_2} \beta^{(t)} \bigotimes_{i=1}^n \^Q^{(t)}_i$, fix a shared lower-bound selector $\ell$ as in \Cref{definition:probability-lower-bound}. Then \text{RecursiveCouple}($1, \alpha, \beta$) in Algorithm \ref{alg:coupling-process} outputs joint random variables $(X,Y)$ such that $X  \sim \^P$ and $Y \sim \^Q$.
\end{lemma}

\Cref{lemma:validity-coupling-mix-prod} follows by a straightforward induction. We defer the proof to Appendix \ref{section:analysis-recursive-coupling}.

\subsubsection{Discrepancy of recursive coupling}
We now analyze the discrepancy of the recursive coupling process. Let $\^P = \sum_{s=1}^{k_1} \alpha^{(s)} \bigotimes_{i=1}^n \^P^{(s)}_i$ and $\^Q = \sum_{t=1}^{k_2} \beta^{(t)} \bigotimes_{i=1}^n \^Q^{(t)}_i$ be two mixtures of product distributions. Define 
\begin{align*}
(X,Y) \gets \text{RecursiveCouple}(1, \alpha, \beta).
\end{align*}
Let $\+C_{\-{RC}}$ denote the joint distribution of $(X,Y)$. By the coupling inequality, the discrepancy $\Pr[\+C_{\-{RC}}]{X\neq Y}$ is always at least $\DTV{\^P}{\^Q}$. The following lemma provides an upper bound on the discrepancy of our recursive coupling.

\begin{lemma}[Discrepancy of recursive coupling]\label{lemma:discrepancy-recursive-coupling}
Let $[q]$ be a finite domain.
For two mixtures of product distributions $\^P = \sum_{s=1}^{k_1} \alpha^{(s)} \bigotimes_{i=1}^n \^P^{(s)}_i$ and $\^Q = \sum_{t=1}^{k_2} \beta^{(t)} \bigotimes_{i=1}^n \^Q^{(t)}_i$ over $[q]^n$, the recursive coupling satisfies
\begin{align*}
\Pr[\+C_{\-{RC}}]{X\neq Y} \leq (4nq)^{k_1+k_2-1} \cdot \DTV{\^P}{\^Q}.
\end{align*}
\end{lemma}

The proof of \Cref{lemma:discrepancy-recursive-coupling} follows the known techniques of~\cite{KwonECM23}. We include the proof in Appendix \ref{sec:discrepancy-RC} for completeness.

\subsection{Implementation of oracles}\label{subsection:implementation-oracles}

In this section, we implement the oracle queries in \Cref{assumption:coupling} based on Algorithm \ref{alg:coupling-process}. Consider the execution of Algorithm \ref{alg:coupling-process} on two mixtures of product distributions $\^P = \sum_{s=1}^{k_1} \alpha^{(s)} \bigotimes_{i=1}^n \^P^{(s)}_i$ and $\^Q = \sum_{t=1}^{k_2} \beta^{(t)} \bigotimes_{i=1}^n \^Q^{(t)}_i$. Define the following set of states and transitions between them.

For each recursion instance $\text{RecursiveCouple}(j, \bar{\alpha}, \bar{\beta})$, we create a state $S$ with label $(j, \bar{\alpha}, \bar{\beta})$, where $1 \leq j \leq n + 1$. A state $S$ is said to be in layer $j$ if $j$ is the first coordinate of the label. 
We add a special failure state $S=\perp$ to denote that the coupling fails.
Let $\+S$ denote the set of all possible states that can be reached by the coupling process with a \emph{positive} probability, including the failure state $\perp$.

For each state $S \in \+S$ with label $(j, \bar{\alpha}, \bar{\beta})$ for $j\in [n]$, we define the following transitions according to the coupling process. The coupling process samples $(c_{\bar{\^P}}, e_{\bar{\^P}})$ and $(c_{\bar{\^Q}}, e_{\bar{\^Q}})$ for $\bar{\^P}$ and $\bar{\^Q}$, respectively, and then makes the transition. 
\begin{itemize}
    \item \textbf{Type-I}: For any $c \in [q]$ with $\ell(j, \bar\alpha, \bar\beta, c) > 0$, with probability $\ell(j, \bar\alpha, \bar\beta, c)$, $S$ moves to a new state $S'$ with label $(j+1, \bar{\alpha}, \bar{\beta})$. Define a transition $t$ from $S$ to $S'$ with weight $w(t) = \ell(j, \bar\alpha, \bar\beta, c)$ and label $L(t) = (c,c)$. This means that $S$ moves to $S'$ through $t$ with probability $\ell(j, \bar\alpha, \bar\beta, c)$ and sets $X_j = Y_j = c$.
    \item \textbf{Type-II}: For any $c \in [q]$ with $p_c \defeq \min\{\bar{\^P}_j(c), \bar{\^Q}_j(c)\}-\ell(j, \bar\alpha, \bar\beta, c) > 0$, with probability $p_c$, $S$ moves to a new state $S'$ with label $(j+1, \bar{\alpha}_{j,c}, \bar{\beta}_{j,c})$. Define a transition $t$ from $S$ to $S'$ with weight $w(t) = p_c$ and label $L(t) = (c,c)$. This means that $S$ moves to $S'$ through $t$ with probability $p_c$ and sets $X_j = Y_j = c$.
    \item \textbf{Type-III}: With the remaining probability, the coupling fails ($X_j \neq Y_j$). For any $c,c'\in [q]$ with $c \neq c'$ and $p_{c,c'}\defeq \Pr{c_{\bar{\^P}} = c \land c_{\bar{\^Q}} = c'} > 0$, define a transition $t$ from $S$ to the failure state $\perp$ with weight $w(t) = p_{c,c'}$ and label $L(t) = (X_{j:n},Y_{j:n})$, where $X_j = c$, $Y_j = c'$, and $X_{j+1:n},Y_{j+1:n}$ are sampled independently in Lines~\ref{line:recursive-sample-1} and~\ref{line:recursive-sample-2} of Algorithm \ref{alg:coupling-process}. Hence, the label contains a pair of random configurations, both of length $n-j + 1$.
    The label $L(t)$ is random because we further sample the pair $(X_{j+1:n},Y_{j+1:n})$.
    We remark that there are multiple transitions from $S$ to the failure state $\perp$ with different weights and labels.
\end{itemize}

\begin{lemma}\label{lem:num-states}
The total number of states is bounded by $|\+S| \leq (nq+1)^{k_1+k_2-1} + 1$.
\end{lemma}

\begin{proof}

    For any non-failed state $S$ (i.e., $S\neq \perp$), recall that it is labeled by
    $(j,\bar\alpha,\bar\beta)$.
    Define the \emph{active component count}
    \[
        m(S)\defeq \bigl|\{s\in[k_1]:\bar\alpha^{(s)}>0\}\bigr|+\bigl|\{t\in[k_2]:\bar\beta^{(t)}>0\}\bigr|.
    \]
    At the root $S_{\mathrm{root}} = (1, \alpha, \beta)$, it holds that $m(S_{\mathrm{root}}) \leq k_1+k_2$.

    Consider any transition from a non-failed state $S$ with label $(j,\bar\alpha,\bar\beta)$ in layer $j$.
    Consider a type-I transition $t$ from $S$ to some state $S_I$, where the label on $t$ is $L(t) = (c,c)$.
    Then $S_I$ keeps the same mixing weights, hence 
    $$m(S_I)=m(S).$$
    Consider a type-II transition from $S$ to some state $S_{II}$ with label $L(t) = (c,c)$. Note that in this case, we must have $\min\{\bar{\^P}_j(c), \bar{\^Q}_j(c)\}-\ell(j, \bar\alpha, \bar\beta, c) > 0$.
    The mixing weights in $S_{II}$ are updated to
    \begin{align*}
      \bar{\alpha}^{(s)}_{j,c}
      &= \bar{\alpha}^{(s)} \cdot \frac{\^P^{(s)}_j(c) - \ell(j, \bar\alpha, \bar\beta, c)}{\bar{\^P}_j(c) - \ell(j, \bar\alpha, \bar\beta, c)}
      \quad \text{for all } s \in [k_1],\\
      \bar{\beta}^{(t)}_{j,c}
      &= \bar{\beta}^{(t)} \cdot \frac{\^Q^{(t)}_j(c) - \ell(j, \bar\alpha, \bar\beta, c)}{\bar{\^Q}_j(c) - \ell(j, \bar\alpha, \bar\beta, c)}
      \quad \text{for all } t \in [k_2].
    \end{align*}
    By the definition of $\ell(j, \bar\alpha, \bar\beta, c)$ in \Cref{definition:probability-lower-bound}, there exists at least one active index $s \in [k_1]$ or $t \in [k_2]$ such that
    ($\^P^{(s)}_j(c)-\ell(j, \bar\alpha, \bar\beta, c)=0$ and $\bar \alpha^{(s)}>0$) or ($\^Q^{(t)}_j(c)-\ell(j, \bar\alpha, \bar\beta, c)=0$ and $\bar \beta^{(t)}>0$).
    Note that the two denominators $\bar{\^P}_j(c) - \ell(j, \bar\alpha, \bar\beta, c)$ and $\bar{\^Q}_j(c) - \ell(j, \bar\alpha, \bar\beta, c)$ are both positive.
    Therefore, whenever the process moves to $S_{II}$, at least one previously-active mixture component becomes inactive, and
    \begin{align} \label{eq:active-component-count-update}
        m(S_{II})\le m(S)-1.
    \end{align}
    Finally, for a type-III transition, the next state is $S_{III} = \perp$, and the whole process terminates. All transitions except type-III transitions move to a new state in the next layer $j+1$.

    We bound the number of states by an encoding argument.
    Now fix a layer $j\in\{1,2,\ldots,n+1\}$.
    Any non-failed state at layer $j$ is encoded by the choices made along a path from the root to that state.
    Along this path, at each previous layer we either take a type-I transition (which we encode by the symbol $0$) or take a type-II transition with label $(c,c)$ for some value $c\in[q]$.
    Hence, given the root state $S_{\mathrm{root}} = (1, \alpha, \beta)$, we can use a string $d$ of length $j-1$ with characters in $[q]\cup\{0\}$ to encode the label of the current state. To decode such a string, start from $S_{\mathrm{root}}$ at layer $1$. For each $h=1,2,\ldots,j-1$, if $d_h = 0$, update the current state $(h,\bar{\alpha},\bar{\beta})$ to $(h+1, \bar{\alpha}, \bar{\beta})$; if $d_h = c$ for some $c\in[q]$, update it to $(h+1, \bar{\alpha}_{h,c}, \bar{\beta}_{h,c})$. Let $r$ be the number of nonzero symbols in the code $d$. Using~\eqref{eq:active-component-count-update}, $m(S_{\mathrm{root}}) \leq k_1+k_2$, and the fact that the two mixing-weight vectors remain valid probability distributions throughout the recursive coupling process (so each vector has at least one active component), we have
    \begin{align*}
    r \leq m(S_{\mathrm{root}}) - 2 = k_1 + k_2 - 2.
    \end{align*}

    We count the number of possible strings $d$ of length $j-1$ that encode a non-failed state at layer $1\leq j \leq n+1$. 
    For a fixed $j$ and $r$, there are at most $\binom{j-1}{r}$ ways to choose the $r$ positions at which a non-zero symbol occurs, and then $q^r$ ways to choose the corresponding nonzero symbols.
    Hence the number of non-failed states over all layers is at most
    \begin{align*}
        \sum_{j=1}^{n+1}\sum_{r=0}^{k_1+k_2-2} \binom{j-1}{r} q^r
        &\le (n+1)\sum_{r=0}^{k_1+k_2-2} (nq)^r
        \le (nq+1)^{k_1+k_2-1}.
    \end{align*}
    Finally, there is one failure state. Hence the total number of states is at most $(nq+1)^{k_1+k_2-1} + 1$.
\end{proof}

The oracle queries in \Cref{assumption:coupling} can be implemented by the following algorithm. One execution of the coupling process can be represented as a random walk on the state space $\+S$. 
The random walk starts from the root state $S = S_{\mathrm{root}} = (1, \alpha, \beta)$.
At each step, the random walk samples a transition $t = (S \to S')$ from $S$ with probability $w(t)$ and moves to the next state $S'$. The random walk ends either at a state in layer $n+1$ (coupling succeeds) or at the failure state $\perp$ (coupling fails). The coupling process outputs a pair $(X_{1:n},Y_{1:n})$, which is determined by the labels of all transitions along the random walk.

\paragraph{Preprocessing step}
We build the DAG with all states in $\+S$ and all transitions between them. Using \Cref{lem:num-states}, the total number of states and transitions in the DAG is at most $O(q^2(nq+1)^{k_1+k_2-1})$.
For each state $S \in \+S$, compute $p_{\text{fail}}(S)$, the probability that a random walk starting from $S$ reaches the failure state $\perp$. By definition,
\begin{align*}
p_{\text{fail}}(\perp) = 1 \text{ and } p_{\text{fail}}(S) = 0 \text{ for all } S \text{ in layer } n+1.
\end{align*}
For each non-failed state $S \in \+S$ in layer $j\in [n]$, we have the following recursive relation:
\begin{align*}
p_{\text{fail}}(S) = \sum_{t =(S \to S') \in \text{transitions from } S} w(t) p_{\text{fail}}(S').
\end{align*}
Using dynamic programming, we can compute $p_{\text{fail}}(S)$ for all $S \in \+S$ in time $O(q^2(nq+1)^{k_1+k_2-1})$.
By the definition of $p_{\text{fail}}(S)$, the failure probability of the recursive coupling is
\begin{align*}
 \Pr[\+C_{\-{RC}}]{X \neq Y} = p_{\text{fail}}(S_{\mathrm{root}}).
\end{align*}

\paragraph{Sampling query}
Next, we give an algorithm that generates a random sample of $X$ conditional on $X \neq Y$ in the recursive coupling. Equivalently, we can sample a random walk on the DAG starting from the root state $S_{\mathrm{root}}$, conditional on the random walk reaching the failure state $\perp$. If $p_{\text{fail}}(S_{\mathrm{root}})=0$, then $\DTV{\^P}{\^Q}=0$ by the coupling inequality, and this sampling query is not needed. Otherwise, the query can be implemented by the following algorithm.
\begin{itemize}
    \item Initialize $S = S_{\mathrm{root}}$.
    \item Sample a random transition $t = (S \to S')$ from $S$ with probability $\frac{w(t)p_{\text{fail}}(S')}{p_{\text{fail}}(S)}$ and then update the current state $S$ to $S'$.
    \item The process terminates when $S$ reaches the failure state $\perp$.
\end{itemize}
The above process gives a random sequence of transitions on the DAG. By concatenating the labels in the sequence, we obtain a random pair $(X,Y) \in [q]^n \times [q]^n$, and we output the sample $X \in [q]^n$.

The correctness of the conditional sample $X$ follows directly from the definition of $p_{\text{fail}}(\cdot)$. The running time to generate a random sample $X$ is $O(q^2(nq+1)^{k_1+k_2-1})$.

\paragraph{Evaluation query}
Given a fixed configuration $\sigma \in [q]^n$, we want to evaluate the probability 
\begin{align*}
    \Pr[\+C_{\-{RC}}]{X = \sigma \land X \neq Y}.
\end{align*}
For each non-failed state $S \in \+S$ in layer $j\in [n]$, let $p^\sigma_{\text{fail}}(S)$ be the probability that a random walk starting from $S$ reaches the failure state $\perp$ and that the label $X_{j:n}$ of the random walk is exactly $\sigma_{j:n}$. By definition, since $S_{\mathrm{root}}$ is the layer-1 state, it holds that
\begin{align*}
    p^\sigma_{\text{fail}}(S_{\mathrm{root}}) = \Pr[\+C_{\-{RC}}]{X = \sigma \land X \neq Y}.
\end{align*}

For each $S \in \+S$ in layer $j = n + 1$, by definition, we have $p^\sigma_{\text{fail}}(S) = 0$.

Fix any non-failed state $S \in \+S$ in layer $j\in [n]$ with $S = (j, \bar{\alpha}, \bar{\beta})$. Define 
\begin{align*}
\+T = \{t \mid t \text{ is a transition from } S \text{ to } S' \neq \perp \text{ and } L(t) = (\sigma_j,\sigma_j)\}.
\end{align*} 
The above definition considers all transitions from $S$ to non-failed states $S'$. We still need to consider the type-III transitions from $S$ to the failure state $\perp$. By definition, there are at most $q(q-1)$ such transitions.
Let $\+T_{\perp}(S)$ denote the set of all type-III transitions from $S$ to the failure state $\perp$.
Every such transition can be specified by a pair $(c,c')$ with $c \neq c'$. We use $t_{c,c'}$ to denote the corresponding type-III transition. Note that the weight $w(t_{c,c'})$ is $\Pr[]{c_{\bar{\^P}} = c \land c_{\bar{\^Q}} = c'}$ in the coupling at Line~\ref{line:couple-random-pair} of Algorithm \ref{alg:coupling-process}. The label of $t_{c,c'}$ is a random pair $(X_{j:n},Y_{j:n})$, where $X_j = c$, $Y_j = c'$, and $X_{j+1:n}$ and $Y_{j+1:n}$ are sampled independently in Lines~\ref{line:recursive-sample-1} and~\ref{line:recursive-sample-2} of Algorithm \ref{alg:coupling-process}. We define
\begin{align*}
\tau(t_{c,c'},\sigma) \defeq \Pr{X_{j+1:n} = \sigma_{j+1:n}} = \sum_{s=1}^{k_1} \bar{\alpha}^{(s)}_{j,c} \prod_{i=j+1}^n \^P^{(s)}_i(\sigma_i).
\end{align*}
Now, we have the following recursive relation:
\begin{align*}
    p^\sigma_{\text{fail}}(S) = \sum_{t=(S\to S') \in \+T} w(t) p^\sigma_{\text{fail}}(S') + \sum_{t_{c,c'} \in \+T_{\perp}(S):c = \sigma_j} w(t_{c,c'}) \tau(t_{c,c'},\sigma).
\end{align*}
Again, we can use dynamic programming on the DAG from layer $n+1$ to layer $1$ to compute all values $p^\sigma_{\text{fail}}(S)$ for all $S \in \+S$ in time $O(q^2(nq+1)^{k_1+k_2-1})$. The final result is $p^\sigma_{\text{fail}}(S_{\mathrm{root}})$.

\section{Algorithm for Mixtures of Boolean Subcubes}\label{sec:algorithm-subcube}

\subsection{Reduction to a counting problem}

Consider two mixtures of Boolean subcubes $\^P = \sum_{s=1}^{k_1} \alpha^{(s)} \bigotimes_{i=1}^n \^P^{(s)}_i$ and $\^Q = \sum_{t=1}^{k_2} \beta^{(t)} \bigotimes_{i=1}^n \^Q^{(t)}_i$. A component $\^P^{(s)}$ over $\{0,1\}^n$ is a Boolean subcube if, for any $i \in [n]$, $\^P^{(s)}_i(1) \in \{0,\frac{1}{2},1\}$. Thus, for any $\omega \in \{0,1\}^n$, $\^P^{(s)}(\omega) = \prod_{i=1}^n \^P^{(s)}_i(\omega_i)$ is a product of values in $\{0,\frac{1}{2},1\}$. We define the following partition $\Lambda^{\text{one}}_{\^P^{(s)}}$, $\Lambda^{\text{zero}}_{\^P^{(s)}}$, and $\Lambda^{\text{half}}_{\^P^{(s)}}$ of $[n]$:
\begin{align*}
\Lambda^{\text{one}}_{\^P^{(s)}} = \{i \in [n] \mid \^P^{(s)}_i(1) = 1\},\,\Lambda^{\text{zero}}_{\^P^{(s)}} = \{i \in [n] \mid \^P^{(s)}_i(1) = 0\},\,\Lambda^{\text{half}}_{\^P^{(s)}}= \left\{i \in [n] \mid \^P^{(s)}_i(1) = \frac{1}{2}\right\}.
\end{align*}
In words, for the product distribution $\^P^{(s)}$, the set $\Lambda^{\text{one}}_{\^P^{(s)}}$ contains all dimensions whose value is fixed to $1$, the set $\Lambda^{\text{zero}}_{\^P^{(s)}}$ contains all dimensions whose value is fixed to $0$, and the set $\Lambda^{\text{half}}_{\^P^{(s)}}$ contains all dimensions whose value is sampled uniformly at random. Similarly, for a component $\^Q^{(t)}$, we define the sets $\Lambda^{\text{one}}_{\^Q^{(t)}}$, $\Lambda^{\text{zero}}_{\^Q^{(t)}}$, and $\Lambda^{\text{half}}_{\^Q^{(t)}}$ in the same way. 

    

The probability of each $\omega \in \{0,1\}^n$ can be written as follows:
\begin{align*}
    \^P^{(s)}(\omega) &= \left(\frac{1}{2}\right)^{\left|\Lambda^{\text{half}}_{\^P^{(s)}}\right|} \cdot \mathbbm{1}{\left[\forall i \in \Lambda^{\text{one}}_{\^P^{(s)}}, \omega_i = 1\right]} \cdot \mathbbm{1}{\left[\forall i \in \Lambda^{\text{zero}}_{\^P^{(s)}}, \omega_i = 0\right]}\\
    \^Q^{(t)}(\omega) &= \left(\frac{1}{2}\right)^{\left|\Lambda^{\text{half}}_{\^Q^{(t)}}\right|} \cdot \mathbbm{1}{\left[\forall i \in \Lambda^{\text{one}}_{\^Q^{(t)}}, \omega_i = 1\right]} \cdot \mathbbm{1}{\left[\forall i \in \Lambda^{\text{zero}}_{\^Q^{(t)}}, \omega_i = 0\right]},
\end{align*}
where the indicator function $\mathbbm{1}{[*]}$ returns $1$ if statement $*$ holds and $0$ otherwise.
For the distribution $\^P^{(s)}$, we say $\omega \in \{0,1\}^n$ is \emph{feasible} if for all $i \in \Lambda^{\text{one}}_{\^P^{(s)}}$, $\omega_i = 1$ and for all $i \in \Lambda^{\text{zero}}_{\^P^{(s)}}$, $\omega_i = 0$. 
Each feasible $\omega \in \{0,1\}^n$ has probability $(1/2)^{|\Lambda^{\text{half}}_{\^P^{(s)}}|}$ under the distribution $\^P^{(s)}$, while each infeasible $\omega \in \{0,1\}^n$ has probability $0$. A similar structural property holds for the distribution $\^Q^{(t)}$.

The feasibility of a configuration $\omega \in \{0,1\}^n$ under the distribution $\^P^{(s)}$ can be expressed by a Boolean formula $F_{\^P^{(s)}}: \{0,1\}^n \to \{0,1\}$ as follows:
\begin{align}\label{eq:boolean-formula-mixtures-of-boolean-subcubes}
    F_{\^P^{(s)}}(\omega) = \bigwedge_{i \in \Lambda^{\text{one}}_{\^P^{(s)}}} \omega_i \land \bigwedge_{i \in \Lambda^{\text{zero}}_{\^P^{(s)}}} \neg \omega_i.
\end{align}
Similarly, for the distribution $\^Q^{(t)}$, define the Boolean formula $F_{\^Q^{(t)}}: \{0,1\}^n \to \{0,1\}$ as follows:
\begin{align}\label{eq:boolean-formula-mixtures-of-boolean-subcubes-2}
    F_{\^Q^{(t)}}(\omega) = \bigwedge_{i \in \Lambda^{\text{one}}_{\^Q^{(t)}}} \omega_i \land \bigwedge_{i \in \Lambda^{\text{zero}}_{\^Q^{(t)}}} \neg \omega_i.
\end{align}
By definition, $\^P^{(s)}(\omega) = (1/2)^{|\Lambda^{\text{half}}_{\^P^{(s)}}|} \cdot F_{\^P^{(s)}}(\omega)$ and $\^Q^{(t)}(\omega) = (1/2)^{|\Lambda^{\text{half}}_{\^Q^{(t)}}|} \cdot F_{\^Q^{(t)}}(\omega)$. The total variation distance between mixtures of Boolean subcubes $\^P = \sum_{s=1}^{k_1} \alpha^{(s)} \bigotimes_{i=1}^n \^P^{(s)}_i$ and $\^Q = \sum_{t=1}^{k_2} \beta^{(t)} \bigotimes_{i=1}^n \^Q^{(t)}_i$ can be expressed as follows:
\begin{align*}
    \DTV{\^P}{\^Q} &= \frac{1}{2} \sum_{\omega \in \{0,1\}^n} \left| \^P(\omega) - \^Q(\omega) \right|  = \frac{1}{2} \sum_{\omega \in \{0,1\}^n} \left| \sum_{s=1}^{k_1} \alpha^{(s)} \prod_{i=1}^n \^P^{(s)}_i(\omega_i) - \sum_{t=1}^{k_2} \beta^{(t)} \prod_{i=1}^n \^Q^{(t)}_i(\omega_i) \right|
    \\ &= \frac{1}{2} \sum_{\omega \in \{0,1\}^n} \left| \sum_{s=1}^{k_1} \alpha^{(s)} \cdot \left(\frac{1}{2}\right)^{\left|\Lambda^{\text{half}}_{\^P^{(s)}}\right|} \cdot F_{\^P^{(s)}}(\omega) - \sum_{t=1}^{k_2} \beta^{(t)} \cdot \left(\frac{1}{2}\right)^{\left|\Lambda^{\text{half}}_{\^Q^{(t)}}\right|} \cdot F_{\^Q^{(t)}}(\omega) \right|.
\end{align*}

Define the \emph{characteristic function} $F:\{0,1\}^n \to \{0,1\}^{k_1+k_2}$ as follows:
\begin{align*}
    F(\omega) = \left(F_{\^P^{(1)}}(\omega), F_{\^P^{(2)}}(\omega), \ldots, F_{\^P^{(k_1)}}(\omega), F_{\^Q^{(1)}}(\omega), F_{\^Q^{(2)}}(\omega), \ldots, F_{\^Q^{(k_2)}}(\omega)\right).
\end{align*}
By the above calculation, the value of $|\^P(\omega) - \^Q(\omega)|$ is fully determined by $F(\omega)$. The key observation is that the image space of $F$ has constant size $2^{k_1+k_2}$, so we can reorganize all $\omega \in \{0,1\}^n$ according to the value of the characteristic function $F(\omega)$. Formally,
\begin{align}\label{eq:dtv-mixtures-of-boolean-subcubes}
    \DTV{\^P}{\^Q} &= \frac{1}{2} \sum_{\chi \in \{0,1\}^{k_1+k_2}} N_{\chi} \cdot \left|  \sum_{s=1}^{k_1} \alpha^{(s)} \cdot \left(\frac{1}{2}\right)^{\left|\Lambda^{\text{half}}_{\^P^{(s)}}\right|} \cdot \chi_s - \sum_{t=1}^{k_2} \beta^{(t)} \cdot \left(\frac{1}{2}\right)^{\left|\Lambda^{\text{half}}_{\^Q^{(t)}}\right|} \cdot \chi_{t+k_1} \right|,
\end{align}
where $N_{\chi}$ counts the number of $\omega \in \{0,1\}^n$ such that $F(\omega) = \chi$, i.e.,
\begin{align*}
N_\chi = \left|\{\omega \in \{0,1\}^n \mid F(\omega) = \chi\}\right|.
\end{align*}
By~\eqref{eq:dtv-mixtures-of-boolean-subcubes}, computing the total variation distance between two mixtures of Boolean subcubes $\^P$ and $\^Q$ reduces to counting $N_{\chi}$ for all $\chi \in \{0,1\}^{k_1+k_2}$.

\subsection{The counting algorithm}

\begin{lemma}\label{lem:counting-algorithm}
Given any $\chi \in \{0,1\}^{k_1+k_2}$, the value of $N_{\chi}$ can be computed in time $O(2^{|S_0|} \cdot n (k_1+k_2))$, where $S_0 = \{\, j \in [k_1 + k_2] \mid \chi_j = 0 \,\}$ is the set of indices $j$ with $\chi_j = 0$.
\end{lemma}

Fix $\chi \in \{0,1\}^{k_1+k_2}$, and consider the set of $\omega \in \{0,1\}^n$ such that $F(\omega) = \chi$. We need to count the number of such $\omega \in \{0,1\}^n$. To simplify the notation, define a sequence of Boolean functions $F_1, F_2, \ldots, F_{k_1+k_2}: \{0,1\}^n \to \{0,1\}$ as follows:
\begin{align*}
\forall j \in [k_1+k_2], \quad F_j(\omega) = \begin{cases}
F_{\^P^{(j)}}(\omega) & \text{if } j \in [k_1], \\
F_{\^Q^{(j-k_1)}}(\omega) & \text{if } j \in [k_1+1, k_1+k_2],
\end{cases}
\end{align*}
where $F_{\^P^{(s)}}$ and $F_{\^Q^{(t)}}$ are defined in~\eqref{eq:boolean-formula-mixtures-of-boolean-subcubes} and~\eqref{eq:boolean-formula-mixtures-of-boolean-subcubes-2}, respectively. For a configuration $\omega \in \{0,1\}^n$ with $F(\omega) = \chi$, we have $F_j(\omega) = \chi_j$ for all $j \in [k_1 + k_2]$.

For a subset $S \subseteq [k_1+k_2]$ of indices, define the following set of configurations:
\begin{align*}
\Phi(S) = \{\omega \in \{0,1\}^n \mid F_j(\omega) = 1 \text{ for all } j \in S\}.
\end{align*}
Let $S_1 = \{j \in [k_1 + k_2] \mid \chi_j = 1\}$ and $S_0 = \{j \in [k_1 + k_2] \mid \chi_j = 0\}$. By the definition of $N_\chi$ and the inclusion-exclusion principle, we have
\begin{align}\label{eq:counting-algorithm}
N_\chi = |\Phi(S_1)| - \left | \{ \omega \in \Phi(S_1) \mid \exists\, j \in S_0, F_j(\omega) = 1 \} \right | = \sum_{S \subseteq S_0} (-1)^{|S|} \cdot |\Phi(S_1 \cup S)|.
\end{align}

\begin{observation}\label{obs:counting-algorithm}
Given any $S \subseteq [k_1+k_2]$, the size $|\Phi(S)|$ can be computed in time $O(n \cdot |S|)$.
\end{observation}

\begin{proof}
By the definition of $F_{\^P^{(s)}}$ and $F_{\^Q^{(t)}}$, we have $F_{\^P^{(s)}}(\omega) = 1$ iff $\omega_i = 1$ for all $i \in \Lambda^{\text{one}}_{\^P^{(s)}}$ and $\omega_i = 0$ for all $i \in \Lambda^{\text{zero}}_{\^P^{(s)}}$. Similarly, $F_{\^Q^{(t)}}(\omega) = 1$ iff $\omega_i = 1$ for all $i \in \Lambda^{\text{one}}_{\^Q^{(t)}}$ and $\omega_i = 0$ for all $i \in \Lambda^{\text{zero}}_{\^Q^{(t)}}$. We can go through all $j \in S$ and fix the values of the corresponding dimensions $i \in [n]$. If there is a contradiction, then $|\Phi(S)|=0$. Otherwise, $|\Phi(S)|=2^r$, where $r$ is the number of dimensions whose values are not fixed. The running time is $O(n \cdot |S|)$.
\end{proof}

We now prove \Cref{lem:counting-algorithm}. Combining~\eqref{eq:counting-algorithm} and~\Cref{obs:counting-algorithm}, we can exactly compute the value of $N_{\chi}$ in time 
\begin{align*}
\sum_{S \subseteq S_0} O(n |S_1 \cup S|) \leq \sum_{0 \leq k \leq |S_0|} \binom{|S_0|}{k} \cdot O(n(k_1+k_2)) = O(2^{|S_0|} \cdot n (k_1+k_2)). 
\end{align*}

\subsection{The total variation distance algorithm}
The algorithm for computing the total variation distance between two mixtures of Boolean subcubes $\^P$ and $\^Q$ is obtained by combining~\eqref{eq:dtv-mixtures-of-boolean-subcubes} and \Cref{lem:counting-algorithm}. Let $K=k_1+k_2$. The overall running time is
\begin{align*}
\sum_{0 \leq m \leq K} \binom{K}{m} \cdot O(2^m \cdot n K) = O(3^K K \cdot n) = O(3^{k_1 + k_2} (k_1 + k_2) \cdot n).
\end{align*}
Here $m$ enumerates the number of indices $j$ with $\chi_j = 0$. This proves the main result in \Cref{thm:tv-subcube}.

\section{Proof of the hardness result}\label{sec:proof-hardness-subcube}
\begin{proof}[Proof of~\Cref{thm:hardness-subcube}]
    We show that if there is a polynomial-time algorithm for computing the TV-distance between two mixtures of Boolean subcubes with $k_1+k_2=\Theta(n)$, then there is a polynomial-time algorithm for $\#\mathsf{3SAT}$, restricted to 3-CNF formulas. 
    Let $$\varphi(x_1, x_2, \ldots, x_r) = C_1 \land C_2 \land \ldots \land C_m$$ be a 3-CNF formula over the used variables $x_1, x_2, \ldots, x_r$, and $m \geq 1$.
    Set $N = \max\{r, m\}$. If $N > r$, we add $N-r$ dummy variables $x_{r+1}, \ldots, x_N$. Let $S$ denote the number of satisfying assignments of $\varphi$ over $x_1, \ldots, x_N$. As the dummy variables do not appear in any clause, the number of satisfying assignments for $\varphi$ is given by
    $$\#\mathsf{SAT}(\varphi) = \frac{S}{2^{N-r}}.$$
    We establish the proof by showing the relation between $S$ and the TV-distance between certain mixtures of Boolean subcubes. 
    Let $n = N+1$. We construct two mixtures of subcubes $\^P$ and $\^Q$ over $\{0,1\}^n$. A point in this space can be represented as $(b, \omega_1, \ldots, \omega_N)$, where $b \in \{0, 1\}$ is an additional bit, and $(\omega_1, \ldots, \omega_N)$ is an assignment to $(x_1, \ldots, x_N)$. For each $j\in[m]$, let $x_{j,a},x_{j,b},x_{j,c}$ be the three variables appearing in the clause $C_j$. Define $(\omega_{j,a}, \omega_{j,b}, \omega_{j,c}) \in \{0, 1\}^3$ to be the unique assignment to $x_{j,a},x_{j,b},x_{j,c}$ that falsifies $C_j$. Let $U_j$ be a subcube such that 
    $$\Pr[U_j]{(b,x_{j,a},x_{j,b},x_{j,c}) = (0, \omega_{j,a}, \omega_{j,b}, \omega_{j,c})} = 1,$$ 
    and the other $N-3$ coordinates are uniform over $\{0,1\}$. Define 
    $$\^P = \frac{1}{m}\sum_{j=1}^m U_j,$$
    is a mixture of $k_1 = m$ subcubes. 
    Next we define $\^Q$. Let $V_0$ and $V_1$ be two subcubes such that $$\Pr[V_0]{b=0} = \Pr[V_1]{b=1} = 1,$$ and that are uniform over $(x_1, \ldots, x_N)$. Let $\lambda = \frac{1}{2m}$, and define $$\^Q = \lambda V_0 + (1-\lambda) V_1,$$ which is a mixture of $k_2 = 2$ subcubes.
    
    We now compute $\DTV{\^P}{\^Q}$. For any assignment $\omega = (\omega_1, \ldots, \omega_N)$, let $$F(\omega) = \left|\left\{ j\in[m]\mid \omega\text{ falsifies }C_j \right\} \right|$$ denote the number of clauses falsified by $\omega$.
    For every $\omega \in \{0,1\}^N$, we have $$\^P(0, \omega) = \frac{1}{m} \sum_{j=1}^m U_j(0,\omega) = \frac{1}{m} \sum_{j=1}^m \mathbf{1}[\omega\text{ falsifies }C_j] \cdot 2^{-(N-3)} = \frac{F(\omega)}{m}\cdot 2^{-(N-3)} = \frac{8F(\omega)}{m}\cdot 2^{-N},$$
    and $\^P(1, \omega) = 0$ by definition of $U_j$. Remark that $$\^Q(0,\omega) = \lambda \cdot 2^{-N} = \frac{1}{2m}\cdot 2^{-N} < \frac{8F(\omega)}{m}\cdot 2^{-N} = \^P(0, \omega) \quad \text{iff } F(\omega) \geq 1;$$ $$\^Q(1,\omega) = (1-\lambda) \cdot 2^{-N} > 0 = \^P(1, \omega).$$
    Therefore
    \begin{align}\label{eq:dtv-3CNF}
        \DTV{\^P}{\^Q} &= \sum_{(b,\omega) \in \{0,1\}^{N+1}} \max\left\{0, \^P(b, \omega) - \^Q(b, \omega)\right\} = \sum_{\omega \in \{0,1\}^N: F(\omega)\geq 1} \^P(0,\omega) - \^Q(0,\omega)\notag\\
        &= \sum_{\omega \in \{0,1\}^N: F(\omega)\geq 1} \frac{8F(\omega)}{m}\cdot 2^{-N} - \sum_{\omega \in \{0,1\}^N: F(\omega)\geq 1} \frac{1}{2m} \cdot 2^{-N}.
    \end{align}

    For the first term, we consider counting all pairs $(\omega,C_j)$ such that $\omega$ falsifies $C_j$. On the one hand, for each $\omega$, it falsifies $F(\omega)$ clauses; on the other hand, for each $C_j$, it is falsified by $2^{N-3}$ assignments. Thus $$\sum_{\omega \in \{0,1\}^N: F(\omega)\geq 1} F(\omega) = \sum_{\omega \in \{0,1\}^N}F(\omega) = m\cdot 2^{N-3}.$$
    
    For the second term, there are $S$ assignments $\omega$ satisfying all clauses, that is, $F(\omega) = 0$. So there are $2^N - S$ assignments satisfying $F(\omega)\geq 1$. Substituting back to~\eqref{eq:dtv-3CNF}, we have
    \begin{align*}
        \DTV{\^P}{\^Q} &= \sum_{\omega \in \{0,1\}^N: F(\omega)\geq 1} \frac{8F(\omega)}{m}\cdot 2^{-N} - \sum_{\omega \in \{0,1\}^N: F(\omega)\geq 1} \frac{1}{2m} \cdot 2^{-N}\\
        &= \frac{2^{-N+3}}{m} \cdot (m\cdot 2^{N-3}) - \frac{2^{-N}}{2m} \cdot (2^N - S) = 1-\frac{1}{2m} + \frac{2^{-N}}{2m}S.
    \end{align*}
    Thus the number of satisfying assignments is given by
    $$\#\mathsf{SAT}(\varphi) = \frac{S}{2^{N-r}} = 2^r(2m\cdot\DTV{\^P}{\^Q}-2m+1).$$

    The constructed distributions live on the $n = N+1$ dimensional subcube, with $k_1 + k_2 = m+2$ total components. Since $N = \max\{r,m\}$ and $r \leq 3m$, we have $m \leq N \leq 3m$, hence $$k_1+k_2 = m+2 = \Theta(N+1) = \Theta(n).$$
    Therefore, an algorithm for~\Cref{problem:compute-subcube} with $k_1 + k_2 = \Theta(n)$ yields an exact algorithm for $\#\mathsf{3SAT}$.
\end{proof}

\ifthenelse{\boolean{DoubleBlind}}{
}{
\section*{Acknowledgements}
Weiming Feng acknowledges the support of ECS grant 27202725 from Hong Kong RGC.
We thank Arnab Bhattacharyya and Guy Van den Broeck for bringing the problem to our attention.
}

\bibliographystyle{alpha}
\bibliography{refs}

\appendix

\section{Validity of Sampling and Coupling Processes}
\label{section:analysis-recursive-coupling}

\subsection{Validity of the sampling process}
\begin{proof}[Proof of~\Cref{lemma:validity-of-recursive-sampling}]
    We prove by induction on $j$ that for any $j\in\{1,2,\ldots,n+1\}$ and any mixing weights $\bar\alpha = (\bar{\alpha}^{(s)})_{s=1}^{k_1}$,
    \text{RecursiveSample}$(j,\bar \alpha)$ outputs a random sample
    \[
        X_{j:n}\sim \bar{\^P} \defeq \sum_{s=1}^{k_1} \bar{\alpha}^{(s)} \bigotimes_{i=j}^n \^P^{(s)}_i.
    \]

    The base case $j=n+1$ is immediate since the algorithm returns $\emptyset$.
    For the induction step, fix any $\omega_{j:n}\in [q]^{n+1-j}$.
    Let $c=\omega_j$.
    In Line~\ref{line:sample-random-pair}, the algorithm chooses $(c,0)$ with probability $\ell^{\^P}(j, \bar{\alpha}, c)$ and chooses $(c,1)$ with probability $\bar{\^P}_j(c)-\ell^{\^P}(j, \bar{\alpha}, c)$.
    Conditioned on $e=0$, by the induction hypothesis the recursive call returns $\omega_{j+1:n}$ with probability
    $\sum_{s=1}^{k_1}\bar{\alpha}^{(s)}\prod_{i=j+1}^n \^P^{(s)}_i(\omega_i)$.
    Conditioned on $e=1$, the algorithm updates the mixing weights to $\bar{\alpha}_{j,c} = (\bar{\alpha}^{(s)}_{j,c})_{s=1}^{k_1}$ and again by the induction hypothesis the recursive call returns $\omega_{j+1:n}$ with probability
    $\sum_{s=1}^{k_1}\bar{\alpha}^{(s)}_{j,c}\prod_{i=j+1}^n \^P^{(s)}_i(\omega_i)$. Therefore,
    \begin{align*}
        \Pr{X_{j:n} =\omega_{j:n}}
        &= \ell^{\^P}(j, \bar{\alpha}, c)\sum_{s=1}^{k_1}\bar{\alpha}^{(s)}\prod_{i=j+1}^n \^P^{(s)}_i(\omega_i)
        + \bigl(\bar{\^P}_j(c)-\ell^{\^P}(j, \bar{\alpha}, c)\bigr)\sum_{s=1}^{k_1}\bar{\alpha}^{(s)}_{j,c}\prod_{i=j+1}^n \^P^{(s)}_i(\omega_i)\\ &= \bar{\^P}(\omega_{j:n}),
    \end{align*}
    where the second equality is exactly the decomposition in~\eqref{eq:P-decomposition}.
    Note that~\eqref{eq:probability-lower-bound-inequality} guarantees $$\bar{\^P}_j(c)-\ell^{\^P}(j, \bar{\alpha}, c) = \sum_{s\in[k_1]: \bar{\alpha}^{(s)} >0} \bar{\alpha}^{(s)} \left(\^P^{(s)}_j(c) - \ell^{\^P}(j, \bar{\alpha}, c)\right) \geq 0,$$
    which completes the induction.
\end{proof}

\subsection{Validity of the coupling}
\begin{proof}[Proof of~\Cref{lemma:validity-coupling-mix-prod}]
    We prove a slightly more general statement: for any $j\in\{1,2,\ldots,n+1\}$ and any mixing-weight vectors $\bar\alpha = (\bar{\alpha}^{(s)})_{s=1}^{k_1}$ and $\bar\beta = (\bar{\beta}^{(t)})_{t=1}^{k_2}$, \text{RecursiveCouple}$(j,\bar{\alpha},\bar{\beta})$ outputs joint random variables $(X_{j:n},Y_{j:n})$ such that
    \[
        X_{j:n}\sim \bar{\^P}\defeq \sum_{s=1}^{k_1}\bar{\alpha}^{(s)}\bigotimes_{i=j}^n \^P^{(s)}_i,
        \qquad
        Y_{j:n}\sim \bar{\^Q}\defeq \sum_{t=1}^{k_2}\bar{\beta}^{(t)}\bigotimes_{i=j}^n \^Q^{(t)}_i.
    \]
    Taking $j=1$ with the initial weights $\alpha = (\alpha^{(s)})_{s=1}^{k_1}$ and $\beta = (\beta^{(t)})_{t=1}^{k_2}$ gives the lemma. We now prove $X_{j:n}\sim\bar{\^P}$ by induction on $j$.
    The proof for $Y_{j:n}\sim\bar{\^Q}$ is symmetric.

    We proceed by induction on $j$. The base case is $j=n+1$.
    The algorithm returns $(\emptyset,\emptyset)$, which matches the (trivial) distributions on the empty product space.

    Fix any $j\le n$. In Line~\ref{line:couple-random-pair} of Algorithm \ref{alg:coupling-process}, the algorithm samples
    $(c_{\bar{\^P}},e_{\bar{\^P}})$ and $(c_{\bar{\^Q}},e_{\bar{\^Q}})$ using the coupling $\bar{\+C}$.
    By construction of $\bar{\+C}$, the marginal distribution of $(c_{\bar{\^P}},e_{\bar{\^P}})$ is exactly as follows: for each $c\in[q]$,
    $\Pr[]{(c_{\bar{\^P}},e_{\bar{\^P}})=(c,0)}=\ell(j, \bar\alpha, \bar\beta, c)$ and 
        $\Pr[]{(c_{\bar{\^P}},e_{\bar{\^P}})=(c,1)}=\bar{\^P}_j(c)-\ell(j, \bar\alpha, \bar\beta, c)$.
    All the above probabilities are non-negative due to~\Cref{definition:probability-lower-bound}.
    Hence, ignoring the other side, the pair-sampling rule on the $\bar{\^P}$ side matches exactly the rule in
    \text{RecursiveSample} (Line~\ref{line:sample-random-pair} of Algorithm~\ref{alg:recursive-sampling}). 
    Now fix any $\omega_{j:n}\in[q]^{n+1-j}$ and write $c=\omega_j$.
    We compute $\Pr{X_{j:n}=\omega_{j:n}}$ by conditioning on $e_{\bar{\^P}}\in\{0,1\}$:
    \begin{align*}
        \Pr{X_{j:n}=\omega_{j:n}}
        &= \Pr{(c_{\bar{\^P}},e_{\bar{\^P}})=(c,0)}\cdot \Pr{X_{j+1:n}=\omega_{j+1:n}\mid (c_{\bar{\^P}},e_{\bar{\^P}})=(c,0)}\\
        &\quad + \Pr{(c_{\bar{\^P}},e_{\bar{\^P}})=(c,1)}\cdot \Pr{X_{j+1:n}=\omega_{j+1:n}\mid (c_{\bar{\^P}},e_{\bar{\^P}})=(c,1)}.
    \end{align*}
    If $e_{\bar{\^P}}=0$, the algorithm recurses with
    \text{RecursiveCouple}$(j+1,\bar{\alpha},\bar{\beta})$; by the induction hypothesis, the marginal distribution of $X_{j+1:n}$ is given by
        $\sum_{s=1}^{k_1}\bar{\alpha}^{(s)}\prod_{i=j+1}^n \^P^{(s)}_i(\omega_i)$.
    If $e_{\bar{\^P}}=1$, the algorithm updates the weights to $\bar{\alpha}_{j,c} = (\bar{\alpha}^{(s)}_{j,c})_{s=1}^{k_1}$.
    If $c_{\bar{\^P}}=c_{\bar{\^Q}}$, then the algorithm recurses via \text{RecursiveCouple}; by the induction hypothesis, the marginal law of $X_{j+1:n} = \omega_{j+1:n}$ is
    $\sum_{s=1}^{k_1}\bar{\alpha}^{(s)}_{j,c}\prod_{i=j+1}^n \^P^{(s)}_i(\omega_i)$.
    Otherwise, the algorithm samples $X_{j+1:n}$ via \text{RecursiveSample}; by \Cref{lemma:validity-of-recursive-sampling}, its law is again
    $\sum_{s=1}^{k_1}\bar{\alpha}^{(s)}_{j,c}\prod_{i=j+1}^n \^P^{(s)}_i(\omega_i)$.
    Plugging these together with the above marginal distribution of $(c_{\bar{\^P}},e_{\bar{\^P}})$ into the decomposition, we obtain $\Pr{X_{j:n}=\omega_{j:n}} = \bar{\^P}(\omega_{j:n})$ by the same argument as in the proof of \Cref{lemma:validity-of-recursive-sampling}.
    Therefore, $X_{j:n}\sim\bar{\^P}$.
\end{proof}

\section{Discrepancy of the Recursive Coupling}\label{sec:discrepancy-RC}
\begin{proof}[Proof of~\Cref{lemma:discrepancy-recursive-coupling}]
    Consider a step $(j, \bar\alpha, \bar\beta)$ in the \text{RecursiveCouple}($j, \bar{\alpha}, \bar{\beta}$) process with $\bar{\^P} = \sum_{s=1}^{k_1}\bar{\alpha}^{(s)}\bigotimes_{i=j}^n \^P^{(s)}_i$ and $\bar{\^Q} = \sum_{t=1}^{k_2}\bar{\beta}^{(t)}\bigotimes_{i=j}^n \^Q^{(t)}_i$. 
    Let $k(\bar{\alpha})$ and $k(\bar{\beta})$ denote the number of nonzero weights in $(\bar{\alpha}^{(s)})_{s=1}^{k_1}$ and $(\bar{\beta}^{(t)})_{t=1}^{k_2}$, respectively. 
    Formally, $k(\bar{\alpha})=|\{\bar{\alpha}^{(s)} > 0 \mid 1\leq s \leq k_1\}|$ and $k(\bar{\beta})$ is defined similarly. 
    For $j, \bar{\alpha}, \bar{\beta}$, we define a parameter $\delta(j, \bar{\alpha}, \bar{\beta})$ as follows:
    \begin{align*}
    \delta\left(j, \bar\alpha, \bar\beta \right) \defeq \max_{ \substack{ \Lambda \in \+S, \omega \in [q]^\Lambda}}\left| \bar{\^P}_\Lambda(\omega) - \bar{\^Q}_\Lambda(\omega) \right| = \max_{ \substack{ \Lambda \in \+S, \omega \in [q]^\Lambda}} \left| \sum_{s=1}^{k_1} \bar{\alpha}^{(s)} \prod_{i\in \Lambda} \^P_i^{(s)}(\omega_i) - \sum_{t=1}^{k_2}\bar{\beta}^{(t)}\prod_{i\in \Lambda} \^Q_i^{(t)}(\omega_i) \right|
    \end{align*}
    where the maximum is taken over all subsets $\Lambda \in \+S$ such that 
    \begin{align}\label{eq:S-definition}
    \+S \defeq \left\{ \Lambda \subseteq \{j, j+1, \ldots, n\} \mid |\Lambda| \leq k(\bar{\alpha})+k(\bar{\beta}) -1 \right\}.
    \end{align}
    Note that both $\bar{\^P}$ and $\bar{\^Q}$ are defined over the indices $\{j, j+1, \ldots, n\}$.
    The definition of $\delta\left(j, \bar\alpha, \bar\beta \right)$ enumerates all possible subsets $\Lambda$ of size at most $k(\bar{\alpha})+k(\bar{\beta}) -1$ and finds the configuration $\omega \in [q]^\Lambda$ that maximizes the absolute difference between $\bar{\^P}_\Lambda(\omega)$ and $\bar{\^Q}_\Lambda(\omega)$. Hence, for every such $\Lambda$, we have $\delta\left(j, \bar\alpha, \bar\beta \right) \leq \DTV{\bar{\^P}_\Lambda}{\bar{\^Q}_\Lambda}$, and by the data processing inequality for total variation distance, we have
    \begin{align}\label{eq:delta-dtv-bound}
    \delta\left(j, \bar\alpha, \bar\beta \right) \leq \DTV{\bar{\^P}}{\bar{\^Q}}.
    \end{align}
    
    The algorithm \text{RecursiveCouple} defines a coupling $\bar{\+C}_{\-{RE}}$ between $\bar{\^P}$ and $\bar{\^Q}$.
    Let $X_{j:n} \sim \bar{\^P}$ and $Y_{j:n} \sim \bar{\^Q}$ denote the output of the coupling algorithm.
    Let
     \[F\left(j, \bar\alpha, \bar\beta \right) \defeq \Pr[\bar{\+C}_{\-{RE}}]{X_{j:n}\neq Y_{j:n}},\qquad\text{where }X_{j:n} \sim \bar{\^P}\text{ and }Y_{j:n} \sim \bar{\^Q}.\] 
     We next prove that for the recursive coupling step $\text{RecursiveCouple}(j, (\bar{\alpha}^{(s)})_{s=1}^{k_1}, (\bar{\beta}^{(t)})_{t=1}^{k_2})$,
    \begin{align}\label{eq:F-recursive-coupling}
    F\left(j, \bar\alpha, \bar\beta \right) \leq \left(4(n-j+1)q\right)^{k(\bar{\alpha})+k(\bar{\beta})-1} \cdot \delta\left(j, \bar\alpha, \bar\beta \right).
    \end{align}
    Note that~\eqref{eq:F-recursive-coupling} implies the lemma by plugging in $j = 1$, $\alpha = \bar{\alpha}$, and $\beta = \bar{\beta}$, and using the fact that $k(\alpha) + k(\beta) \leq  k_1 + k_2$ and $\delta(1, \alpha, \beta) \leq \DTV{\^P}{\^Q}$ by~\eqref{eq:delta-dtv-bound}.

    We apply mathematical induction on $j$ from $n+1$ to $1$. When $j=n+1$, \text{RecursiveCouple} always returns $(X_{j:n},Y_{j:n}) = (\emptyset, \emptyset)$, which implies $F\left(n+1, \bar{\alpha}, \bar{\beta}\right) = 0$, and~\eqref{eq:F-recursive-coupling} holds trivially.
    For the inductive step, fix $1 \leq j \leq n$ and we assume that~\eqref{eq:F-recursive-coupling} holds for all $j'$ with $j+1\leq j'\leq n+1$. We need to show that~\eqref{eq:F-recursive-coupling} holds for $j$. We can decompose the probability as follows:
    \begin{align*}
        F\left(j, \bar\alpha, \bar\beta \right) = \Pr[\bar{\+C}_{\-{RE}}]{X_{j:n}\neq Y_{j:n}} &= \sum_{c\in [q]} \Pr[]{(c_{\bar{\^P}},e_{\bar{\^P}}) = (c,0)} \Pr[\bar{\+C}_{\-{RE}}]{X_{j:n}\neq Y_{j:n} \mid (c_{\bar{\^P}},e_{\bar{\^P}}) = (c,0)} \\
        & \, + \sum_{c\in [q]} \Pr[]{(c_{\bar{\^P}},e_{\bar{\^P}}) = (c,1)} \Pr[\bar{\+C}_{\-{RE}}]{X_{j:n}\neq Y_{j:n} \mid (c_{\bar{\^P}},e_{\bar{\^P}}) = (c,1)}.
    \end{align*}
    In the above equation, we partition the probability by enumerating all possible values of $c_{\bar{\^P}} \in [q]$ and $e_{\bar{\^P}} \in \{0,1\}$. The equation follows from the law of total probability. The event $(c_{\bar{\^P}},e_{\bar{\^P}}) = (c,0)$ happens with probability $\ell(j, \bar\alpha, \bar\beta, c)$. By the definition of the coupling, $c_{\bar{\^Q}} = c$ and $e_{\bar{\^Q}} = 0$ in this case, and the algorithm proceeds recursively to Line~\ref{line:recursive-couple-0}. For the case $(c_{\bar{\^P}},e_{\bar{\^P}})=(c,1)$, we can further divide it into two subcases: (1) $c_{\bar{\^P}} = c_{\bar{\^Q}} = c$, which happens with probability $\min\{\bar{\^P}_j(c), \bar{\^Q}_j(c)\} - \ell(j, \bar\alpha, \bar\beta, c)$, and the algorithm proceeds recursively to Line~\ref{line:recursive-couple-1}; (2) $c_{\bar{\^P}} \neq c_{\bar{\^Q}}$, which happens with probability $\max\{\bar{\^P}_j(c)-\bar{\^Q}_j(c), 0\}$, and the coupling fails ($X_j \neq Y_j$ implies $X \neq Y$). Overall, we have the following recursion:
    \begin{align*}
        F\left(j, \bar\alpha, \bar\beta \right) &\leq \sum_{c\in [q]} \ell(j, \bar\alpha, \bar\beta, c) \cdot F\left(j+1, \bar{\alpha}, \bar{\beta}\right)\\ 
        &\quad + \sum_{c\in [q]} \left(\min\{\bar{\^P}_j(c), \bar{\^Q}_j(c)\} - \ell(j, \bar\alpha, \bar\beta, c)\right) \cdot F\left(j+1, \bar{\alpha}_{j,c}, \bar{\beta}_{j,c} \right)\\
        &\quad + \sum_{c\in [q]} \max\{\bar{\^P}_j(c)-\bar{\^Q}_j(c),0\}.
    \end{align*}
    We analyze the three terms on the right-hand side one by one. By the induction hypothesis, the first term is at most $\left(4(n-j)q\right)^{k(\bar{\alpha})+k(\bar{\beta})-1} \cdot \delta\left(j+1, \bar\alpha, \bar\beta \right)$. For the second term, the induction hypothesis gives
    $F\left(j+1, \bar{\alpha}_{j,c}, \bar{\beta}_{j,c} \right) \leq \left(4(n-j)q\right)^{k(\bar{\alpha}_{j,c})+k(\bar{\beta}_{j,c})-1} \cdot \delta\left(j+1, \bar{\alpha}_{j,c}, \bar{\beta}_{j,c} \right)$. Moreover, to upper bound $F(j, \bar{\alpha}, \bar{\beta})$, we only consider $c\in [q]$ with $\min\{\bar{\^P}_j(c), \bar{\^Q}_j(c)\} - \ell(j, \bar\alpha, \bar\beta, c) \neq 0$. The third term is the total variation distance between $\bar{\^P}_j$ and $\bar{\^Q}_j$. By the definition of $\delta\left(j, \bar{\alpha}, \bar{\beta}\right)$, taking $\Lambda = \{j\}$ and $\omega_j = c$ implies $|\bar{\^P}_j(\omega) - \bar{\^Q}_j(\omega)| \leq \delta\left(j, \bar{\alpha}, \bar{\beta}\right)$. Hence, the third term is at most $q \cdot \delta\left(j, \bar{\alpha}, \bar{\beta}\right)$. Overall, we have the following recursion:
    \begin{align*}
        F\left(j, \bar\alpha, \bar\beta \right) &\leq \sum_{c\in [q]} \ell(j, \bar\alpha, \bar\beta, c) \cdot \left(4(n-j)q\right)^{k(\bar{\alpha})+k(\bar{\beta})-1} \cdot \delta\left(j+1, \bar\alpha, \bar\beta \right)\\
        &\quad + \sum_{c \in [q]} \left(\min\{\bar{\^P}_j(c), \bar{\^Q}_j(c)\} - \ell(j, \bar\alpha, \bar\beta, c)\right) \cdot \left(4(n-j)q\right)^{k(\bar{\alpha}_{j,c})+k(\bar{\beta}_{j,c})-1} \delta\left(j+1, \bar{\alpha}_{j,c}, \bar{\beta}_{j,c} \right)\\
        &\quad + q \cdot \delta\left(j, \bar{\alpha}, \bar{\beta}\right).
    \end{align*}
    Recall that the updated parameters when $\min\{\bar{\^P}_j(c), \bar{\^Q}_j(c)\} - \ell(j, \bar\alpha, \bar\beta, c) > 0$:
    \begin{align*}
        \bar\alpha_{j,c}^{(s)} \defeq \bar\alpha^{(s)} \frac{\^P^{(s)}_j(c) - \ell\left(j, \bar\alpha, \bar\beta , c\right)}{\bar{\^P}_j(c) - \ell\left(j, \bar\alpha, \bar\beta, c\right)} \text{ and } \bar\beta_{j,c}^{(t)} \defeq \bar\beta^{(t)} \frac{\^Q^{(t)}_j(c) - \ell\left(j, \bar\alpha, \bar\beta , c\right)}{\bar{\^Q}_j(c) - \ell\left(j, \bar\alpha, \bar\beta, c\right)}.
    \end{align*}
    Remark that if $\bar{\alpha}^{(s)}=0$ or $\bar{\beta}^{(t)}=0$, then the updated mixing weight $\bar{\alpha}^{(s)}_{j,c}=0$ or $\bar{\beta}^{(t)}_{j,c} = 0$, that is $k(\bar{\alpha}_{j,c}) \leq k(\bar{\alpha})$ and $k(\bar{\beta}_{j,c}) \leq k(\bar{\beta})$. By~\Cref{definition:probability-lower-bound}, there must exist some $s \in [k_1]$ with $\bar{\alpha}^{(s)} > 0$ or $t \in [k_2]$ with $\bar{\beta}^{(t)} > 0$, such that ${\^P}^{(s)}_j(c)= \ell(j, \bar\alpha, \bar\beta, c)$ or ${\^Q}^{(t)}_j(c) = \ell(j, \bar\alpha, \bar\beta, c)$, and thus there must exist $\bar{\alpha}^{(s)} > \bar\alpha_{j,c}^{(s)} = 0$ or $\bar{\beta}^{(t)}>\bar{\beta}_{j,c}^{(t)} = 0$.  Hence, if the recursion goes to Line~\ref{line:recursive-couple-1} with $c_{\bar{\^P}} = c_{\bar{\^Q}} = c$, then it must hold that $\min\{\bar{\^P}_j(c), \bar{\^Q}_j(c)\} - \ell(j, \bar\alpha, \bar\beta, c) > 0$ and we have $k(\bar{\alpha}_{j,c}) + k(\bar{\beta}_{j,c}) \leq k(\bar{\alpha}) + k(\bar{\beta})-1$. Next, we claim the following inequality:
    \begin{align}\label{eq:delta-recursive-coupling}
  \left(\min\{\bar{\^P}_j(c), \bar{\^Q}_j(c)\} - \ell(j, \bar\alpha, \bar\beta, c)\right) \delta\left(j+1, \bar{\alpha}_{j,c}, \bar{\beta}_{j,c} \right) \leq 3 \delta\left(j,\bar{\alpha},\bar{\beta}\right).
    \end{align}
    Let us first assume that~\eqref{eq:delta-recursive-coupling} holds. Then we have
    \begin{align*}
        F\left(j, \bar\alpha, \bar\beta \right) &\leq \left(4(n-j)q\right)^{k(\bar{\alpha})+k(\bar{\beta})-1} \cdot \delta\left(j+1, \bar\alpha, \bar\beta \right) + 3q\left(4(n-j)q\right)^{k(\bar{\alpha})+k(\bar{\beta})-2} \delta\left(j, \bar{\alpha}, \bar{\beta} \right) + q\delta\left(j, \bar{\alpha}, \bar{\beta}\right)\\
        &\leq \left(4(n-j+1)q\right)^{k(\bar{\alpha})+k(\bar{\beta})-1} \delta\left(j, \bar{\alpha}, \bar{\beta}\right),
    \end{align*}
    where the last inequality holds because $\delta(j+1,\bar\alpha, \bar \beta) \leq \delta(j,\bar\alpha,\bar\beta)$. 
    This finishes the induction step.

    Finally, we prove~\eqref{eq:delta-recursive-coupling}. Note that~\eqref{eq:delta-recursive-coupling} holds trivially when $\min\{\bar{\^P}_j(c), \bar{\^Q}_j(c)\} - \ell(j, \bar\alpha, \bar\beta, c) = 0$. We assume the case $\ell(j, \bar\alpha, \bar\beta, c) < \bar{\^P}_j(c) \leq \bar{\^Q}_j(c)$, and the other case, when $\ell(j, \bar\alpha, \bar\beta, c) < \bar{\^Q}_j(c) \leq \bar{\^P}_j(c)$, can be proved similarly.
    By the definition of $\delta(j+1, \bar{\alpha}_{j,c}, \bar{\beta}_{j,c})$ and~\eqref{eq:S-definition}, one needs to enumerate all possible subsets $\Lambda' \in \{j+1, \ldots, n\}$ with size at most $k(\bar{\alpha}_{j,c}) + k(\bar{\beta}_{j,c}) - 1$ and all configurations $\omega' \in [q]^{\Lambda'}$. Note that $k(\bar{\alpha}_{j,c}) + k(\bar{\beta}_{j,c}) \leq k(\bar{\alpha}) + k(\bar{\beta})-1$. We consider all $\Lambda'$ with size at most $k(\bar{\alpha}) + k(\bar{\beta}) - 2$, which covers all possible subsets $\Lambda'$ in the definition of $\delta(j+1, \bar{\alpha}_{j,c}, \bar{\beta}_{j,c})$. 
    Hence, we can get an upper bound on $\delta(j+1, \bar{\alpha}_{j,c}, \bar{\beta}_{j,c})$.
    Formally, for any $\Lambda' \subseteq \{j+1, \ldots, n\}$, $|\Lambda'| \leq k(\bar{\alpha}) + k(\bar{\beta}) - 2$, and any $\omega' \in [q]^{\Lambda'}$, we have
    \begin{align}
        &\left| \sum_{s=1}^{k_1} \bar{\alpha}^{(s)}_{j,c} \prod_{i\in \Lambda'} \^P_i^{(s)}(\omega'_i) - \sum_{t=1}^{k_2}\bar{\beta}^{(t)}_{j,c}\prod_{i\in \Lambda'} \^Q_i^{(t)}(\omega'_i) \right|\notag\\ 
        = &\left| \sum_{s=1}^{k_1} \bar{\alpha}^{(s)} \frac{{\^P}^{(s)}_j(c) - \ell(j, \bar\alpha, \bar\beta, c)}{\bar{\^P}_j(c) - \ell(j, \bar\alpha, \bar\beta, c)} \prod_{i\in \Lambda'} \^P_i^{(s)}(\omega'_i) - \sum_{t=1}^{k_2}\bar{\beta}^{(t)} \frac{{\^Q}^{(t)}_j(c) - \ell(j, \bar\alpha, \bar\beta, c)}{\bar{\^Q}_j(c) - \ell(j, \bar\alpha, \bar\beta, c)} \prod_{i\in \Lambda'} \^Q_i^{(t)}(\omega'_i) \right|\notag\\
        (\ast)\quad \leq & \frac{1}{\bar{\^P}_j(c) - \ell(j, \bar\alpha, \bar\beta, c)} \left| \sum_{s=1}^{k_1} \bar{\alpha}^{(s)} \^P^{(s)}_j(c) \prod_{i\in \Lambda'} \^P_i^{(s)}(\omega'_i) - \sum_{t=1}^{k_2}\bar{\beta}^{(t)} \^Q^{(t)}_j(c) \prod_{i\in \Lambda'} \^Q_i^{(t)}(\omega'_i) \right|\label{eq:pr-fail-3}\\
        &+ \frac{\ell(j, \bar\alpha, \bar\beta, c)}{\bar{\^P}_j(c) - \ell(j, \bar\alpha, \bar\beta, c)} \left| \sum_{s=1}^{k_1} \bar{\alpha}^{(s)} \prod_{i\in \Lambda'} \^P_i^{(s)}(\omega'_i) - \sum_{t=1}^{k_2}\bar{\beta}^{(t)} \prod_{i\in \Lambda'} \^Q_i^{(t)}(\omega'_i) \right|\label{eq:pr-fail-4}\\
        &+ \left| \frac{1}{\bar{\^P}_j(c) - \ell(j, \bar\alpha, \bar\beta, c)} - \frac{1}{\bar{\^Q}_j(c) - \ell(j, \bar\alpha, \bar\beta, c)} \right| \sum_{t=1}^{k_2}\bar{\beta}^{(t)} \left(\^Q^{(t)}_j(c) - \ell(j, \bar\alpha, \bar\beta, c)\right) \prod_{i\in \Lambda'} \^Q_i^{(t)}(\omega'_i)\label{eq:pr-fail-5}
    \end{align} 
    The inequality ($\ast$) is obtained by first inserting two terms, $\sum_{t=1}^{k_2}\bar{\beta}^{(t)} \frac{{\^Q}^{(t)}_j(c) - \ell(j, \bar\alpha, \bar\beta, c)}{\bar{\^P}_j(c) - \ell(j, \bar\alpha, \bar\beta, c)} \prod_{i\in \Lambda'} \^Q_i^{(t)}(\omega'_i)$ and $-\sum_{t=1}^{k_2}\bar{\beta}^{(t)} \frac{{\^Q}^{(t)}_j(c) - \ell(j, \bar\alpha, \bar\beta, c)}{\bar{\^P}_j(c) - \ell(j, \bar\alpha, \bar\beta, c)} \prod_{i\in \Lambda'} \^Q_i^{(t)}(\omega'_i)$, whose sum is $0$, and then applying the triangle inequality.
    Our task is now reduced to bounding the three terms in~\eqref{eq:pr-fail-3}, \eqref{eq:pr-fail-4}, and \eqref{eq:pr-fail-5}.
    For \eqref{eq:pr-fail-3}, consider the definition of $\delta(j, \bar{\alpha}, \bar{\beta})$. We can take $\Lambda = \{j\} \cup \Lambda'$ and $\omega = (\omega_j \gets c) \cup \omega'$ in the definition. Note that $\Lambda$ is allowed in the definition of $\delta(j,\bar\alpha,\bar\beta)$ because $|\Lambda'| \leq k(\bar{\alpha}) + k(\bar{\beta}) - 2$. Then
    \[\frac{1}{\bar{\^P}_j(c) - \ell(j, \bar\alpha, \bar\beta, c)}\left| \sum_{s=1}^{k_1} \bar{\alpha}^{(s)} \^P^{(s)}_j(c) \prod_{i\in \Lambda'} \^P_i^{(s)}(\omega'_i) - \sum_{t=1}^{k_2}\bar{\beta}^{(t)} \^Q^{(t)}_j(c) \prod_{i\in \Lambda'} \^Q_i^{(t)}(\omega'_i) \right| \leq \frac{\delta\left(j,\bar{\alpha},\bar{\beta}\right)}{\bar{\^P}_j(c) - \ell(j, \bar\alpha, \bar\beta, c)}.\]
    For \eqref{eq:pr-fail-4}, as $\ell(j, \bar\alpha, \bar\beta, c) \leq 1$, consider $\Lambda = \Lambda'$ and $\omega = \omega'$ in the definition of $\delta(j, \bar\alpha, \bar\beta)$; we have 
    \begin{align*}
    \frac{\ell(j, \bar\alpha, \bar\beta, c)}{\bar{\^P}_j(c) - \ell(j, \bar\alpha, \bar\beta, c)} \left| \sum_{s=1}^{k_1} \bar{\alpha}^{(s)} \prod_{i\in \Lambda'} \^P_i^{(s)}(\omega'_i) - \sum_{t=1}^{k_2}\bar{\beta}^{(t)} \prod_{i\in \Lambda'} \^Q_i^{(t)}(\omega'_i) \right| &\leq \frac{\ell(j, \bar\alpha, \bar\beta, c) \cdot \delta\left(j, \bar\alpha, \bar\beta\right)}{\bar{\^P}_j(c) - \ell(j, \bar\alpha, \bar\beta, c)}\\ 
    &\leq \frac{\delta\left(j, \bar\alpha, \bar\beta\right)}{\bar{\^P}_j(c) - \ell(j, \bar\alpha, \bar\beta, c)}.
    \end{align*}
    For \eqref{eq:pr-fail-5}, as $\^Q_i^{(t)}(\omega'_i) \leq 1$ and $\^Q_j^{(t)}(c) \geq \ell(j, \bar\alpha, \bar\beta, c)$ for $t \in [k_2]$ with $\bar\beta^{(t)} > 0$, 
    \[\sum_{t=1}^{k_2}\bar{\beta}^{(t)} \left(\^Q^{(t)}_j(c) - \ell(j, \bar\alpha, \bar\beta, c)\right) \prod_{i\in \Lambda'} \^Q_i^{(t)}(\omega'_i) \leq \sum_{t=1}^{k_2}\bar{\beta}^{(t)} \left(\^Q^{(t)}_j(c) - \ell(j, \bar\alpha, \bar\beta, c)\right) = \bar{\^Q}_j(c) - \ell(j, \bar\alpha, \bar\beta, c),\]
    thus
    \begin{align*}
        &\left| \frac{1}{\bar{\^P}_j(c) - \ell(j, \bar\alpha, \bar\beta, c)} - \frac{1}{\bar{\^Q}_j(c) - \ell(j, \bar\alpha, \bar\beta, c)} \right| \sum_{t=1}^{k_2}\bar{\beta}^{(t)} \left(\^Q^{(t)}_j(c) - \ell(j, \bar\alpha, \bar\beta, c)\right) \prod_{i\in \Lambda'} \^Q_i^{(t)}(\omega'_i)\\ 
        \leq &\left| \frac{1}{\bar{\^P}_j(c) - \ell(j, \bar\alpha, \bar\beta, c)} - \frac{1}{\bar{\^Q}_j(c) - \ell(j, \bar\alpha, \bar\beta, c)} \right| \left( \bar{\^Q}_j(c) - \ell(j, \bar\alpha, \bar\beta, c) \right)\\=&\left| \frac{\bar{\^P}_j(c) - \bar{\^Q}_j(c)}{\bar{\^P}_j(c) - \ell(j, \bar\alpha, \bar\beta, c)} \right|\leq \frac{\delta\left(j, \bar\alpha, \bar\beta\right)}{\bar{\^P}_j(c) - \ell(j, \bar\alpha, \bar\beta, c)}.
    \end{align*}
    Then \eqref{eq:delta-recursive-coupling} follows by combining the three inequalities above. 
\end{proof}

\end{document}